%% file: main.tex
\documentclass[12pt, twoside]{report}
\usepackage[margin = 1in, a4paper, width = 150mm, top = 25mm, bottom = 25mm, bindingoffset = 6mm]{geometry}
\usepackage{amsfonts, amsmath, amssymb, amsthm, mathtools, bm}
\usepackage[none]{hyphenat}
\usepackage{fancyhdr}
\usepackage{graphicx}
\graphicspath{{images/}}
\usepackage{ragged2e}
\usepackage{ulem}
\usepackage{pdfpages}
\usepackage{bm}
\usepackage{amsfonts}
\usepackage{braket}
\usepackage{tikz-cd}
\usepackage{hyperref}
\usepackage{fancyhdr}
\usepackage{caption}
\usepackage{subcaption}
\usepackage[style = alphabetic, sorting = none]{biblatex}
\usepackage{graphicx}
\usepackage{tikz}
\usepackage{caption}
\usepackage{tabularx}
\usepackage{graphicx,subcaption}
\usepackage[utf8]{inputenc} 
\usepackage[T1]{fontenc}

\usepackage{amsthm}
\newtheorem{theorem}{Theorem}[section]

\newtheorem{lemma}[theorem]{Lemma}
\theoremstyle{definition}
\newtheorem{definition}{Definition}[section]
\newtheorem{example}{Example}[section]
\newtheorem{notation}{Notation}[section]
\newtheorem{question}{Question}[section]
\newtheorem*{theorem*}{Theorem}

\def\R{{\mathbb R}}

\def\Z{{\mathbb Z}}

\def\N{{\mathbb N}}

\def\T{{\mathcal T}}

\def\bigO{\mathcal{O}}
\def\E{\mathbb{E}}

\title{
    {Convergence Properties of the Asynchronous Maximum Model}\\
    {\large Australian National University}\\
}
\author{John Larkin}
\date{Tuesday 4 June 2024}

\begin{document}
\maketitle

\chapter*{Abstract}
Let $G = (V,E)$ be a connected directed graph on $n$ vertices. Assign values from the set $\{1,2,\dots,n\}$ to the vertices of $G$ and update the values according to the following rule: uniformly at random choose a vertex and update its value to the maximum of the values in its neighbourhood. The value at this vertex can potentially decrease. This random process is called the asynchronous maximum model. Repeating this process we show that for a strongly connected directed graph eventually all vertices have the same value and the model is said to have \textit{converged}. In the undirected case the expected convergence time is shown to be asymptotically (as $n\to \infty$) in $\Omega(n\log n)$ and $\bigO(n^2)$ and these bounds are tight. We further characterise the convergence time in $\bigO(\frac{n}{\phi}\log n)$ where $\phi$ is the vertex expansion of $G$. This provides a better upper bound for a large class of graphs. Further, we show the number of rounds until convergence is in $\bigO((\frac{n}{\phi}\log n)g(n))$ with high probability, where $g(n)$ satisfies $\frac{1}{g^2(n)} \to 0$ as $n \to \infty$.
\par
For a strongly connected directed graph the convergence time is shown to be in $\bigO(nb^2 + \frac{n}{\phi'}\log n)$ where $b$ is a parameter measuring directed cycle length and $\phi'$ is a parameter measuring vertex expansion.

\chapter*{Acknowledgements}
I would like to express a sincere thanks to my supervisor Dr. Ahad N. Zehmakan. I am thankful for the ideas, expertise and direction he has provided for this project. I am appreciative of his additions to the introduction section of the paper. I am grateful for his knowledge and guidance throughout this project, which has helped me to grow in my enjoyment of research.

\tableofcontents

\chapter{Introduction}

\input{chapters/introduction}

\chapter{Preliminaries}
\input{chapters/background}

\chapter{Undirected Graphs}
\input{chapters/asynchronousmaximummodel}

\chapter{Strongly Connected Graphs}

\input{chapters/directed}

\chapter{Conclusion}

\input{chapters/conclusion}

\RaggedRight\sloppy\printbibliography

\chapter{Appendix}
\input{chapters/appendix}

\end{document}

%% file: chapters/introduction.tex
Interactions between entities are ubiquitous in the modern world. Consider a social network; the entities are people and the interaction is whether they are friends or not. A virus spreading; the entities are again people and the interaction is whether they have spent time together or not. Fire propagation; the entities are areas likely to burn and the interaction is some method whereby fire can spread between them.

\par From a mathematical perspective, cf.~\cite{wolfram2018cellular,adler1991bootstrap,zehmakan2019spread}, it is natural to abstractly model this in the form of a graph $G$ which consists of vertices (entities) and edges (interactions between them). The advantage is now we can in some sense `forget' about the underlying model, whether it is a social network, virus or fire, and just consider the abstract graph. Then results for one model can generalise to others.

\par To model various process such as opinion formation, information spreading, fire propagation, and virus transmission, one needs to define a state for each agent/vertex. This state can then update as a result of interaction with the connections and by following a certain deterministic or stochastic updating rule. For example, consider coloring the vertices of a graph either blue or white. Then, an update consists of each vertex changing its current colour to the colour most common in its neighbourhood and keeping its colour in the case of a tie. This is called the \textit{Majority Model}~\cite{peleg2002local,gartner2018majority} and could be interpreted as people (vertices) having certain beliefs (blue or white) which change according to the most popular belief out of all their friends (neighbours in the graph). One may study this model on certain classes of graphs and indeed this has been done in the case of undirected cycles $C_n$ \cite[][]{RMM} and Erd\H{o}s–R\'{e}nyi random graphs $\mathcal{G}_{n,p}$ \cite{OFinER}. For $\mathcal{G}_{n,p}$, it was shown if the vertices are colored blue independently with probability $p_b \leq \frac{1}{2} - \omega(\frac{1}{\sqrt{np}})$ where $(1+\varepsilon)\frac{\log n}{n} \leq p$ the majority model becomes fully red almost asymptotically surely (i.e., with a probability tending to 1). This is an example of a threshold behaviour where if the initial density of blue nodes in the graph is slightly less than $\frac{1}{2}$ the process becomes fully red. Importantly, the conversion of the problem to a graph-theoretic setting allows for rigorous mathematical analysis and the statement of such results.

Several other extensions of the Majority Model have been considered such as in a noisy setup~\cite{balister2010random}, with biased tie-breaking rule~\cite{gartner2021majority}, the presence of random edges~\cite{out2021majority}, and with different threshold values~\cite{dreyer2009irreversible}.

\par Further questions can also be studied. For example what is the minimum set of vertices you need to colour blue so that the entire graph almost surely becomes all blue? \cite{OFinER} This corresponds to the people in the community you should persuade if you want your opinion to spread. In terms of fire propagation; what is the optimal location for a firebreak to prevent the spread of wildfire? The answer to this was shown in general to be NP-Hard to compute, so its difficulty can be quantified.\cite{firebreak}. \newline

\par The models presented so far can be broadly classified into the field of Agent-based modelling. This area of research is used in but not limited to the study of social sciences, biology, environmental sciences, business and networks. The purpose of an agent based model is generally to try and capture the large-scale behaviour of a system when potentially only small scale interactions are known.~\cite{abm_cite}.

\par One could create an Agent-based model in order to study the formation and change of individual opinions in the real-world. With some assumptions the problem can be translated into a mathematical context which can be precisely studied. The Majority Model is one such example. Although the formation of one's opinion is more complex than choosing the most popular opinion out of their friends, the simplification allows for rigorous mathematical analysis.

\par This paper will focus on one particular dynamic graph model; the asynchronous maximum model. Here, each vertex is given some initial value, for example an integer between $1$ and $n$. Then some vertex is chosen uniformly at random and its value is updated to be the maximum of its neighbours. Crucially, the chosen vertex \textit{only} considers the values of its neighbours, so potentially its own value could decrease. The term asynchronous arises since the vertices update one at a time. 

\par Arguably, the following two questions are the most well-studied questions about various dynamic graph process, which model real-world applications, cf.~\cite{goles1983iterative,frischknecht2013convergence,n2020rumor}.

\begin{enumerate}
    \item (Period) Does the process ever terminate or reach a stable cycle of states? If so how many states are in the cycle?
    \item (Convergence) How many rounds on average does it take to reach stability?
\end{enumerate}

The study of the period and convergence time for a strongly connected directed graph $G$ forms a large part of this paper. In the case where $G$ is undirected we can provide bounds on the convergence time and constructions that achieve these bounds. Although there is no immediate translation of this problem to a scenario in the real world it does not make the study of such a process redundant. Rather we hope the theoretical understanding and proof techniques gained could potentially be applied to other interesting problems. \newline

\par 
We briefly introduce some of the proof techniques used. Many of the arguments rely on the observation that if you have two geometrically distributed random variables $X, Y$ with success probabilities $p$ and $q$ with $p > q$, then $\E[X] \leq \E[Y]$. In fact, we generally use this argument when $Y$ is a geometric distribution and $X$ is some process where the probability of success changes but is always greater than or equal to $q$.
\par 
Observe that there are at most $n^n$ ways to assign values from the set $\{1,2\dots,n\}$ to a graph $G$ on $n$ vertices. One can create a graph $\mathcal{G}$ where every vertex is one of these $n^n$ valuations and draw a directed edge from $G_1$ to $G_2$ when there is a non-zero probability of $G_1$ updating to $G_2$ under the asynchronous maximum model. This is called the Markov Chain of Possibilities for $G$. A walk in this graph corresponds to a sequence of updates in the asynchronous maximum model. Further, the Markov Chain can be partitioned into maximal strongly connected components. There are certain maximal strongly connected components that act as `sinks' in the Markov Chain and as the number of rounds increases we almost surely become trapped in a sink. The largest sink is the period of the process and the convergence time is the expected number of rounds (in the worst case) taken to reach a sink.

\par In this way the problems of period and convergence can be studied by analysing the Markov Chain of Possibilities. Another crucial observation (in the undirected case) is that once an edge in $G$ exists with the current maximum value in the graph on both vertices it can never be removed by any update. This provides a measure to determine how many rounds until convergence, we can instead examine how many rounds until every vertex in the graph is part of one of these edges.

\par These techniques allow us to provide bounds on the expected number of rounds until the process converges and also the type of graphs that the model converges to.

\newpage
\section{Summary of Results}

\par 
In Chapter 2 we define the notation of a valuation, a functions assigning numbers to vertices of $G = (V,E)$. We define $k$-boundaries for a set $S \subseteq V$ which formalises the notion of a breadth first search from $S$. We then define an iterative graph model which generalises the notion of updating values on the vertices of $G$. By considering all possible valuations and the probability to transition between them we obtain a Markov chain. This leads to the definition of period of convergence for \textit{any} iterative graph model. Then we introduce the asynchronous maximum model. The convergence time results are summarised in Table \ref{table:1}.
\par In Chapter 3 we analyse the period and convergence on the more simple (for this problem) class of undirected graphs. The period of the process is shown to be 1 while the convergence time $\T(G)$ is bounded between $\Omega(n\log n)$ and $\bigO(n^2)$ for any undirected graph $G$. Further we provide the bound of $\bigO(\frac{n}{\phi}\log n)$ where $\phi$ is the vertex expansion of $G$, potentially giving a much better upper bound than $\bigO(n^2)$ for a large class of graphs with good expansion properties. Further, we show that the convergence time is concentrated around $\bigO(\frac{n}{\phi}\log n)$ in the sense that with high probability the model will converge in $\bigO(g(n)\frac{n}{\phi}\log n)$ where $g(n)$ is any function with the property that $\frac{1}{g^2(n)} \to 0$ as $n\to \infty$. Therefore, almost all undirected graphs on $n$ vertices will converge in $\bigO(g(n)\frac{n}{\phi}\log n)$ rounds.

\par In Chapter 4 we generalise to strongly connected directed graphs and provide period and convergence time theorems. The period is shown to be 1. The convergence time is in $\bigO(nb^2 + \frac{n}{\phi '}\log n )$ where $b$ is a parameter relating to cycle length in $G$ and $\phi'$ relates to the vertex expansion.

\begin{table}[h!]
\centering
\begin{tabular}{ |c| c | c | }
\hline
 Convergence Time $\mathcal{T}(G)$ & Graph Type  & Reference \\ 
 \hline
   $\Omega(n \log n) $  & Undirected & \ref{asymptotics_G} \\ 
   $\bigO(n^2) $  & Undirected & \ref{asymptotics_G}  \\ 
  $\Theta(n \log n)$  & Undirected $K_n$ & \ref{complete_graph_time}  \\ 
 $\Theta(n^2)$  & Undirected $P_n$ & \ref{path_graph_time} \\
 $\bigO(\frac{n}{\phi}\log n)$  & Undirected, $\phi = $ vertex expansion & \ref{asymptotics_vertex}  \\
 $\bigO(nb^2 + \frac{n}{\phi'}\log n)$  & Strongly connected. Parameters: $b,\phi'$ & \ref{conv_time_sc_graph}  \\
 \hline
\end{tabular}
\caption{Convergence Time $\mathcal{T}(G)$ Summary}
\label{table:1}
\end{table}

%% file: chapters/background.tex
For completeness we define a directed graph. We assume the reader is familiar with concepts such as a directed graphs, walks and (strong and weakly) connected. All asymptotics hold with respect to $n \to \infty$.

\begin{definition}
    (Directed and Undirected Graph) A directed graph is a pair $G = (V,E)$ of sets where $E \subseteq V \times V$. Elements of the set $V$ are called vertices (or nodes) and elements of the set $E$ are called edges. The edges are ordered pairs from $V \times V$. A graph is \textit{undirected} if $(u,v) \in E$ whenever $(v,u) \in E$. A graph is \textit{simple} if $(v,v) \notin E$ for all $v \in V$. \cite[][]{diestel}
\end{definition}

Please note when writing $G = (V,E)$ (or simply $G$) in this paper, we will \textbf{always} be referring to a simple and weakly connected graph on $|V| = n$ vertices. However not all graphs we consider are simple and weakly connected on $n$ vertices. For example the Markov Chain of Possibilities $\mathcal{G}$ (which is derived from $G$) will often not be simple or weakly connected and will have $n^n$ vertices. We always use the simple font for $G$ while the fancy font for the Markov Chain $\mathcal{G}$.

\section{Notation}
Below we introduce some important notation which will be used throughout the paper.
\begin{itemize}
    \item $G = (V,E)$ denotes a simple (directed) graph with vertex set $V$ and edge set $E$. $(v_1,v_2) \in E$ denotes a directed edge from $v_1$ to $v_2$. Alternatively we write $v_1 \sim v_2$ to mean there is a directed edge from $v_1$ to $v_2$.
    \item Let $V' \subseteq V$ be a subset of the vertices of a graph $G$. The graph induced on vertex set $V'$ is denoted by $G[V'] = (V', E')$ where $E' = \{(u,v)\in E | \text { both } u, v \in V'\}$.
    \item $K_n$ denotes the (undirected) complete graph on $n$ vertices.
    \item $C_n$ is a directed cycle on $n$ vertices.
    \item $P_n$ is the undirected path graph on $n$ vertices.
    \item For a graph $G = (V,E)$ and $u,v \in V$. Write $u \sim v$ to mean $u$ is adjacent to $v$. That is $(u,v) \in E$ is an edge in $G$.
    \item $\N$ is the set of natural numbers. i.e $\N = \{1,2,3,\dots \}$
    \item $[n] = \{1,2,\dots, n\}$ denotes the set of the first $n$ natural numbers. $[f]$ denotes an equivalence class of functions. It is clear from context whether a number or a function is being discussed.
    \item $\bigsqcup$ denotes a union of disjoint sets.
    \item $\log$ denotes the natural logarithm.
\end{itemize}

\section{Valuations}
We introduce the notion of a valuation. This captures the idea of assigning values such as integers to the vertices of a graph and having these integers change according to some predefined rule.

\begin{definition} (Valuation)
Let $G = (V,E)$ be a graph (directed or undirected) with vertex set $V$ and edge set $E$. Fix $t \in \Z$ such that $t \geq 0$. A \textit{valuation} of the graph $G$ is a function $f_t : V \to [n]$. That is for each vertex $v$ in the graph, we assign to it an element in the set $[n]$.
\end{definition}

We remark that $f_t$ can only take values in $[n]$. This definition seems somewhat restrictive. For example we could allow $f_t:V \to \R$. However one can show that for the asynchronous maximum model we do not lose any generality when the image of the valuation function is $[n]$. For further analysis on this topic please see Appendix \ref{val_app}.

\begin{definition} (Constant Valuation)
    The valuation $f_{con} : V \to [n]$ is called a constant valuation if for all $v \in V$,
    $$f_{con}(v) = k$$
    for some $k \in [n]$.
\end{definition}

\begin{example}
    Consider $G = K_n$, the complete graph on $n$ vertices. The function $f_t(v) = 1$ for all $v \in V$ is a constant valuation. That is every vertex is assigned the integer $1$.
\end{example}

A valuation $f_t$ does not necessarily depend on $t$. The use of the variable $t$ is meant to suggest that the valuation can change with time according to some predefined rule. We will introduce and study one such rule later which will define the asynchronous maximum model. 

\begin{notation}
    $(G,f)$ denotes graph $G$ with valuation $f$.
\end{notation}

\begin{definition}
    Let $G$ be a graph. The \textit{valuation family of $G$} is the set of all valuations of $G$ and is denoted by
    $$\mathcal{F}_G = \{f \, | f: V \to [n] \text{ is a valuation }\}$$
    When the graph is clear from context, $\mathcal{F}_G = \mathcal{F}$.
\end{definition}
Note the family $\mathcal{F}$ is finite since there are at most $n^n$ functions from $V\to [n]$.

\begin{definition}\label{boundary}(Boundary of $S$)
    Let $G$ be a graph and $S \subseteq V$ be non-empty. Define
    $$\Gamma (S) = \{v \in V | u \in S, v \notin S \text{ and } u \sim v\}$$
    we call $\Gamma(S)$ the boundary of $S$. It is the set of neighbours of elements in $S$ that are themselves not in $S$.
\end{definition}

If the graph is strongly connected then we can extend the above definition to partition it into these boundary components.

\begin{definition}\label{nbound}
    ($k$-boundaries of $S$)
    Let $G$ be a graph and $S \subseteq V$ be non-empty. For $k \in \Z, k\geq 0$ define
    $$\Gamma^k (S) = 
    \begin{cases}
         S &  \text{if $k = 0$} \\
         \Gamma(S) & \text{if $k = 1$} \\
         \Gamma(\Gamma^{k-1}(S)) \setminus \Gamma^{k-2}(S) & \text{if $k \geq 2$} \\
    \end{cases}$$
\end{definition}
\begin{notation}
    If the graph $G = (V,E)$ is not clear from context, then for $\emptyset \subsetneq S \subseteq V$, $\Gamma_G^k(S)$ denotes the $k$-boundary of $S$ with respect to $G$ and $S$.
\end{notation}

Informally, $\Gamma^k(S)$ are all vertices that can be reached from $S$ in a walk of length $k$ but not in a walk of length less than $k$. For $v \in V$, with a slight abuse of notation, $\Gamma(v) = \{u \in V | (v,u) \in E\ \text{ and } u \neq v\}$ will be called the \textbf{out-neighbours} of $v$. The $k$-boundaries partition the graph into disjoint sets. Please see Figures \ref{period_ex} and \ref{sc_example_period} for an example of graphs partitioned into $k$-boundaries.

\begin{definition}
(Strongly Connected) Let $G$ be a graph. $G$ is strongly connected if for any two vertices $u,v \in V$ there exists a walk from $u$ to $v$.
\end{definition}

\begin{lemma}\label{nbounpartition}
    Let $G$ be a strongly connected directed graph. Let $\emptyset \subsetneq S \subseteq  V$. There exists a $N \in \N \cup \{0\}$ such that $$V =  \bigsqcup_{k=0}^N \Gamma^k(S) $$
    where $\Gamma^k(S) \neq \emptyset$ for all $0 \leq k \leq N$ and $\Gamma^{k}(S) = \emptyset$ for all $k > N$. That is the $k$-boundaries of $S$ partition the set $V$ into exactly $N+1$ disjoint subsets.
\end{lemma}
We require $G$ to be strongly connected since otherwise there could exist some vertices which cannot be reached by a walk from $S$.

\begin{proof}
    Please see Section \ref{n_boundaries_app} in the Appendix.
\end{proof}

\section{Maximal SCC}
Given a directed graph there is a natural way to partition the vertices into disjoint sets, called the Maximal Strongly Connected Components (SCC).

\begin{definition} (Maximal Strongly Connected Component)
    Let $U \subseteq V$ be a subset of the vertices. $G[U]$ is a maximal strongly connected component if $G[U]$ is a strongly connected graph and for all $v \in V\setminus U$, $G[U \cup \{v\}]$ is not strongly connected.
\end{definition}
Given a directed graph we can partition its vertex set into maximal strongly connected components. The advantage is now each component can be analysed separately, which may be easier to do since the graph is strongly connected. For an example of such a partition please see Figure \ref{mc_example} where every box surrounds a maximal strongly connected component.

\section{Iterative Graph Model}

\begin{definition}
    (Iterative Graph Model) Let $G$ be a graph and $\mathcal{F}$ the valuation family of $G$. For all $f,g \in \mathcal{F}$ assign a probability $p_{(f,g)} \in [0,1]$ of $f$ transitioning to $g$. Further, for all $f \in \mathcal{F}$, we have $\sum_{g \in \mathcal{F}} p_{(f,g)} = 1$. This assignment of probabilities is called an iterative graph model.
\end{definition}
The above definition captures that valuations $f$ are updated to other valuations $g$ according to some probability distribution (which depends on $f$). In practice when defining an iterative graph model we will usually not assign a specific numerical value to every possible transition $f$ to $g$. Instead we can define how to update every valuation $f$ where potentially the update includes some `randomness'. The asynchronous maximum model and the majority model (where colors are in correspondence with numbers) are both examples of iterative graph models. The former includes randomness while the latter does not.

\section{Markov Chain of Possibilities}

\par
Fix a graph $G$ and let $\mathcal{V} = \mathcal{F} $  denote the valuation family of $G$. We now define a directed graph $\mathcal{G} = (\mathcal{V}, \mathcal{E})$. Let $(f,g) \in \mathcal{E}$ if there exists a non-zero probability that $f$ will update to $g$ in the iterative graph model. $(f,g)$ means there is a directed edge from $f$ to $g$. Under this definition self-loops are allowed, so there could be a non-zero probability that a valuation $f$ transitions to itself.

\begin{definition}
    Let $G$ be a graph and $\mathcal{G}$ be given as above. By the definition of an iterative graph model, for every $(f,g) \in \mathcal{E}$ there is a associated non-zero probability of transitioning from $f$ to $g$. The \textit{Markov Chain of Possibilities} is $\mathcal{G}$ with this assignment of probabilities to the edges.
\end{definition}

Please refer to Figure \ref{mc_example}. This is a Markov Chain of Possibilities for some iterative graph model. Every vertex is a valuation and every edge represents some non-zero probability of transitioning from one valuation to another.

\begin{figure}[h!]
\centering
    \begin{tabularx}{0.93\textwidth}{*{2}{>{\centering\arraybackslash}X}}
\begin{tikzpicture}[
every edge/.style = {->, draw=black,very thick},
 vrtx/.style args = {#1/#2}{%
      circle, draw, thick, fill=white,
      minimum size=9mm, label=#1:#2}
                    ]
\node(f) [vrtx= center/$f$] at (0, 0) {};
\node(g) [vrtx= center/$g$] at (2, 0) {};
\node(h) [vrtx= center/$h$] at (4, 0) {};
\node(a) [vrtx= center/$a$] at (6, 0) {};
\node(b) [vrtx= center/$b$] at (2, 2) {};
\node(c) [vrtx= center/$c$] at (4, 2) {};
\node(d) [vrtx = center/$d$] at (6, 2) {};
\node(e) [vrtx = center/$e$] at (4, 4) {};
\path   (f) edge node[yshift = 5mm] {$\frac{1}{2}$} (b);
\path   (b) edge node[yshift = 5mm] {$\frac{2}{3}$} (e);
\path   (b) edge node[xshift = 3mm] {$\frac{1}{3}$} (g);
\path   (g) edge node[yshift = -5mm] {$\frac{1}{4}$} (f);
\path   (g) edge node[yshift = -5mm] {$\frac{1}{8}$} (h);
\path   (h) edge node[xshift = -3mm] {$\frac{1}{2}$} (c);
\path   (h) edge node[yshift = 5mm] {$\frac{1}{2}$} (d);
\path   (d) edge node[xshift = 3mm] {$1$} (a);
\path   (a) edge node[yshift = -5mm] {$\frac{1}{3}$} (h);
\path   (c) edge node[yshift = 5mm] {$1$} (d);

\path   (f) edge [loop above] node {$\frac{1}{2}$} (f);
\path   (g) edge [loop below] node {$\frac{5}{8}$} (g);
\path   (e) edge [loop above] node {$1$} (e);
\path   (a) edge [loop below] node {$\frac{2}{3}$} (a);

\end{tikzpicture} 
\caption*{$\mathcal{G}$}
&   
\begin{tikzpicture}[
every edge/.style = {->, draw=black,very thick},
 vrtx/.style args = {#1/#2}{%
      circle, draw, thick, fill=white,
      minimum size=9mm, label=#1:#2}
                    ]
\node(f) [vrtx= center/$f$] at (0, 0) {};
\node(g) [vrtx= center/$g$] at (2, 0) {};
\node(h) [vrtx= center/$h$] at (4, 0) {};
\node(a) [vrtx= center/$a$] at (6, 0) {};
\node(b) [vrtx= center/$b$] at (2, 2) {};
\node(c) [vrtx= center/$c$] at (4, 2) {};
\node(d) [vrtx = center/$d$] at (6, 2) {};
\node(e) [vrtx = center/$e$] at (4, 4) {};
\path   (f) edge node[yshift = 5mm] {$\frac{1}{2}$} (b);
\path   (b) edge node[yshift = 5mm] {$\frac{2}{3}$} (e);
\path   (b) edge node[xshift = 3mm] {$\frac{1}{3}$} (g);
\path   (g) edge node[yshift = -5mm] {$\frac{1}{4}$} (f);
\path   (g) edge node[yshift = -5mm] {$\frac{1}{8}$} (h);
\path   (h) edge node[xshift = -3mm] {$\frac{1}{2}$} (c);
\path   (h) edge node[yshift = 5mm] {$\frac{1}{2}$} (d);
\path   (d) edge node[xshift = 3mm] {$1$} (a);
\path   (a) edge node[yshift = -5mm] {$\frac{1}{3}$} (h);
\path   (c) edge node[yshift = 5mm] {$1$} (d);

\path   (f) edge [loop above] node {$\frac{1}{2}$} (f);
\path   (g) edge [loop below] node {$\frac{5}{8}$} (g);
\path   (e) edge [loop above] node {$1$} (e);
\path   (a) edge [loop below] node {$\frac{2}{3}$} (a);

\draw [red, dotted, very thick] (-0.6,-1.8) rectangle (2.7,2.8);

\draw [blue, dotted, very thick] (3.3,-1.8) rectangle (6.8,2.8);

\draw [blue, dotted, very thick] (3.3,3) rectangle (4.7,5.7);

\end{tikzpicture} 
\caption*{Partition into Maximal SCC}  
    \end{tabularx}
\caption{Markov Chain $\mathcal{G}$ and the partition into Maximal SCC}
\label{mc_example}
\end{figure}

\begin{definition}\label{absorbing_state_def}
    (Absorbing Components) Let $\mathcal{G}$ be a Markov chain of possibilities. Partition $\mathcal{G}$ into the maximal strongly connected components. The absorbing components are the strongly connected components that have no directed edges leaving them.
\end{definition}

In Figure \ref{mc_example} the absorbing components are boxed in blue. Consider starting at any vertex in $\mathcal{G}$. Since the Markov chain is finite, if we transition between vertices in $\mathcal{G}$ according to the probabilities given by the directed edges, then the probability of being in an absorbing state tends to $1$ as the number of transitions increases. Therefore the asynchronous maximum model will eventually reach one of these absorbing states and never leave it.

\par With regards to the asynchronous maximum model, and update can be regarded as transitioning between vertices of $\mathcal{G}$ and the probability of this update is given by the associated probability on the directed edge. Now we can define both the period and convergence time of this process.

\section{Period}
\begin{definition}\label{period_def} (Period)
    Let $G$ be a graph equipped with an iterative graph model. The $\textit{period}$ of $G$ is the size of the largest absorbing component in the Markov chain of possibilities $\mathcal{G}$.
\end{definition}

\section{Convergence}

\begin{definition}\label{conv_time_G_f}
    (Convergence Time of $(G,f)$) Let $\mathcal{G} = (\mathcal{V},\mathcal{E})$ denote the Markov chain of possibilities for $G$. We can choose a valuation $f \in \mathcal{V}$ and consider a random walk on $\mathcal{G}$ starting at $f$. The transition probabilities in the random walk are the probabilities assigned to the edges in $\mathcal{E}$. During the random walk a transition along an edge will be called an \textbf{iteration} or \textbf{round} or \textbf{update}. Let $Y_f$ be the random variable denoting the number of transitions in the random walk starting at $f$ until a valuation in an absorbing component is reached.
    
    The convergence time of $(G,f)$ is defined as
    $$T(G,f) := \mathbb{E}[Y_f].$$
    
\end{definition}

We are interested in the expected number of rounds in the worst case, which motivates the following definition.
\begin{definition}\label{conv_time}
    (Convergence Time of G)
    The convergence time of $G$ is given by
    $$\mathcal{T}(G) := \max_{f \in [\mathcal F]}T(G,f)$$
    We take the maximum expected number of rounds until an absorbing component is reached over all possible valuations $f$ of $G$.

\end{definition}

\section{Inequalities with Expectations}

\begin{theorem}\label{stoc_dom} 
    (Stochastically Dominated) 
    Let $X,X'$ be random variables and $u : \R \to \R$ a non-decreasing function. Suppose that for all $x \in \R$,
    $$\mathbb{P}(X' \leq x) \leq \mathbb{P}(X \leq x)$$
    then it follows
    $$\E[u(X)]\leq \E[u(X')]$$
    and we say $X'$ is \textit{stochastically dominated} by $X$ or that $X$ \textit{stochastically dominates} $X'$. \cite{book_eco}
\end{theorem}

\begin{theorem}\label{coupling_proof}
    Let $Q_j$ for $j\geq 1$ be a Bernoulli random variable with success probability $q \neq 1$. That is
    $$Q_j = \begin{cases}
        1 & \text{with probability } q \\
        0 & \text{with probability } 1-q 
    \end{cases}$$
    Suppose that for all $j \geq 1$, $q\leq p_j \leq 1$. Let $P_j$ for $j\geq 1$ be a random variable defined by
    $$P_j = \begin{cases}
        1 & \text{with probability } p_j \\
        0 & \text{with probability } 1-p_j \\
    \end{cases}$$
    Define the following random variables.
    $$Q = \min\{j | Q_j = 1\}$$
    $$P = \min\{j | P_j = 1\}$$
    Then for all $k \in \N$, $\mathbb{P}(Q \leq k) \leq \mathbb{P}(P \leq k)$. In other words, $Q$ is stochastically dominated by $P$.
\end{theorem}
The above Theorem will allow us to make the following argument. Suppose we have a random process (called Process 1) composed of rounds where in each round there is a success or failure. For Process 1 the probability of success in each round is $q$. These are modelled by $Q_j$ in the theorem where $Q_j = 1$ is a success. Suppose we have another random process (called Process 2) where the success in round $j$ is $p_j \geq q$. These are modelled by $P_j$ where $P_j = 1$ is a success. Let the random variables $Q, P$ denote the number of rounds until a success for Process 1 and 2 respectively. Intuitively, since $p_j \geq q$ we should expect that $\mathbb{P}(Q \leq k) \leq \mathbb{P}(P \leq k) $. We are more likely to obtain up to $k$ successes in Process 2 than Process 1 since the probability of success in each round is higher for Process 2. This along with Theorem \ref{stoc_dom} then gives us that we expect a success in Process 2 \textit{before} a success in Process 1, again aligning with intuition.
\begin{proof}
    Please see Appendix \ref{coupling_proc_app}. This is a potential proof idea and not a complete proof. The `potential proof' may require more attention to be a complete proof.
\end{proof}

\section{Vertex Expansion}

\begin{definition}
    ($\widetilde{G}$: Dual of $G$) Let $G = (V,E)$ be a directed graph. Define $\widetilde{G} = (V,\widetilde{E})$ such that
    $$\widetilde{E} = \{(v,u) | (u,v) \in E\}.$$
    That is reverse every directed edge in $G$.
\end{definition}
Note that for an undirected graph $G = \widetilde{G}$.

\begin{definition}\label{vertex_exp_def_undirected}
 (Outward and Inward vertex expansion of $G$)
    Let $G$ be a directed graph and $A \subseteq V$ a non empty subset of the vertices. Define
    $$\phi_{out}(G) = \min_{0< |A| \leq \frac{n}{2}} \frac{|\Gamma_{G}(A)|}{|A|}.$$
    Where the minimum is taken over all possible choices of $A$.
    We call $\phi_{out}(G)$ the outward vertex expansion. Similarly define the inward vertex expansion as
    
    $$\phi_{in}(G) := \phi_{out}(\widetilde{G})$$
    and write $\phi_{out} := \phi_{out}(G)$ and $\phi_{in} = \phi_{in}(G)$ when $G$ is clear from context. For an undirected graph $G = \widetilde{G}$ so we have $\phi(G) := \phi_{out} = \phi_{in}$. We write $\phi, \phi_{in}, \phi_{out}$ when the graph $G$ is clear from context. We only write $\phi$ when discussing an undirected graph.
\end{definition}
The quantities $\phi_{out}$ and $\phi_{in}$ are measures of how many out (respectively in) neighbours any set $A \subseteq V$ will have in the worst case when $A$ is not to large. For an undirected graph, if the quantity $\phi$ is large then we expect the graph to be well-connected in some sense. If $\phi = 0$ then the graph is disconnected.

\section{Asynchronous Maximum Model}
Now we introduce the update rule for the asynchronous maximum model for a graph $G = (V,E)$. This is an example of an iterative graph model. Note that we do not explicitly give numerical values for the transition probabilities from $f_t$ to $f_{t+1}$, however the model implicitly defines them.
\newline
\textbf{Update Rule for the Asynchronous Maximum Model}
\par 
    Let $f_0 : V \to C$ denote an initial valuation function for a graph $G$. The valuation function $f_t$ for $t > 0$ is updated according to the following three steps:
    \begin{enumerate}
        \item Choose a vertex $v' \in V$ uniformly at random.
        \item Update the valuation function according to the following rule
        \begin{equation*}
            f_t(v) = \begin{cases}
                \text{max}_{v \sim u}\{f_{t-1}(u)\} & \text{if } v = v' \text{ and } \Gamma(v') \neq \emptyset \\
                f_{t-1}(v) & \text{otherwise }\\
            \end{cases}
        \end{equation*}
        \item Repeat steps 1 and 2.
    \end{enumerate}
In words; we first pick a vertex $v'$ uniformly at random and then update its value to the maximum value of its out-neighbours. If the vertex has no out neighbour its value does not change. All other vertices remain unchanged. We call steps 1 and 2 an \textbf{iteration} or \textbf{round} or \textbf{update} (which is consistent with Definition \ref{conv_time_G_f}). If the graph is clear from context then we may refer to $f_t$ as the graph with valuation $f_t$.

\begin{example}

\begin{figure}[h]
\centering
    \begin{tabularx}{0.8\textwidth}{*{2}{>{\centering\arraybackslash}X}}
\begin{tikzpicture}[
every edge/.style = {draw=black,very thick},
 vrtx/.style args = {#1/#2}{%
      circle, draw, thick, fill=white,
      minimum size=8mm, label=#1:#2}
                    ]
\node(A) [vrtx= center/5] at (0, 0) {};
\node(B) [vrtx= center/5] at (0, -2) {};
\node(C) [vrtx= center/3] at (2,0) {};
\node(D) [vrtx= center/4] at (2,-2) {};
\node(E) [vrtx= center/5] at (4, 0) {};
\node(F) [vrtx= center/2] at (4, -2) {};
\path   (A) edge (B)
        (A) edge (C)
        (C) edge (D)
        (B) edge (D)
        (C) edge (A)
        (C) edge (E)
        (E) edge (F)
        (D) edge (F)
        (E) edge (D);
\end{tikzpicture}
    \caption*{Before update: $f_0$}  
    &   
\begin{tikzpicture}[
every edge/.style = {draw=black,very thick},
 vrtx/.style args = {#1/#2}{%
      circle, draw, thick, fill=white,
      minimum size=8mm, label=#1:#2}
                    ]
\node(A) [vrtx= center/5] at (0, 0) {};
\node(B) [vrtx= center/5] at (0, -2) {};
\node(C) [vrtx= center/3] at (2,0) {};
\node(D) [vrtx= center/4] at (2,-2) {};
\node(E) [red][vrtx= center/4] at (4, 0) {};
\node(F) [vrtx= center/2] at (4, -2) {};
\path   (A) edge (B)
        (A) edge (C)
        (C) edge (D)
        (B) edge (D)
        (C) edge (A)
        (C) edge (E)
        (E) edge (F)
        (D) edge (F)
        (E) edge (D);
\end{tikzpicture}
\caption*{After update: $f_1$}  
    \end{tabularx}
\caption{The top-right vertex is chosen for update}
\label{graphs_ex_val}
\end{figure}

Please refer to Figure \ref{graphs_ex_val}. This is an undirected graph with $n = 6$ vertices and initial valuation $f_0 : V \to [6]$ (left). A vertex is selected uniformly at random which is marked in red. Its value is changed to $\max\{2,3,4\}$. After the update the valuation is $f_1$ (right).
\end{example}

%% file: chapters/asynchronousmaximummodel.tex
In this chapter we analyse the convergence time and period of the asynchronous maximum model on an \textit{undirected} connected graph. Firstly, we show that the period of the model is 1. For any $G$ the convergence time is shown to be between $\Omega(n \log n)$ and $\bigO(n^2)$. Further these bounds are in fact tight for the complete graph $K_n$ and the path graph $P_n$ respectively. Finally we characterise the convergence time in terms of the vertex expansion $\phi$ and show that the number of rounds until convergence cannot be much larger than the expectation.

\par 

\section{Period}\label{period_sec}

Here we show that the period of the asynchronous maximum model is 1. This means eventually the model reaches some valuation (not necessarily unique) where any update does not change the valuation.

\par

\begin{definition} (Maximum in round $t$)
Let $(G,f)$ be a graph with valuation $f$. Define
    $$M_f = \max \{f(v) | v \in V\}. $$
    So $M_f$ is the current maximum value in valuation $f$. The the valuations are indexed according to $t$, for example $f_t$, then we may also refer to $M_{f_t}$ as $M_t$.
\end{definition}

The following lemma will show that from any valuation there exists a sequence of vertices that can be chosen to reach a constant valuation. To highlight the argument in the proof we first provide an illustrative example. Please refer to Figure \ref{period_ex}. This is a graph on $n = 10$ vertices with initial valuation $f_0$. Let $S$ be the set of vertices with the maximum value in the graph under $f_0$. Now partition the graph into the $k$-boundaries $S,\Gamma(S)$ and $\Gamma^2(S)$. The vertices we select are the following; in each round choose a vertex from $\Gamma(S)$. Repeat until we have selected all vertices in $\Gamma(S)$. This will result in all vertices in $S \sqcup \Gamma(S)$ having the maximum value. After this, in each round, we choose vertices from $\Gamma^2(S)$. This will result in all vertices from $S \sqcup \Gamma(S) \sqcup \Gamma^2(S)$ having the maximum value. We have reached a constant valuation and are done. This argument is generalised in the following lemma.

\begin{figure}[h!]
\centering
\begin{tikzpicture}[
every edge/.style = {draw=black,very thick},
 vrtx/.style args = {#1/#2}{%
      circle, draw, thick, fill=white,
      minimum size=8mm, label=#1:#2}]
\node(A) [vrtx= center/10] at (0, 0) {};
\node(B) [vrtx= center/9] at (2, 2) {};
\node(C) [vrtx= center/1] at (2, 0) {};
\node(D) [vrtx= center/3] at (2, -2) {};
\node(E) [vrtx= center/2] at (4, 3) {};
\node(F) [vrtx= center/4] at (4, 2) {};
\node(G) [vrtx= center/6] at (4, 1) {};
\node(H) [vrtx= center/5] at (4, 0) {};
\node(I) [vrtx= center/7] at (4, -1) {};
\node(J) [vrtx= center/8] at (4, -2) {};

\path   (A) edge (B)
        (B) edge (C)
        (A) edge (C)
        (A) edge (D)
        (B) edge (E)
        (B) edge (F)
        (B) edge (G)
        (C) edge (G)
        (C) edge (H)
        (H) edge (I)
        (I) edge (J)
        (C) edge (I)
        (D) edge (J);

\draw [red, very thick] (-0.7,-1) rectangle (0.7,1);
\node[] at (0,1.4) (name) {$S$};

\draw [red, very thick] (1.3,-2.8) rectangle (2.7,2.6);
\node[] at (2,2.9) (name) {$\Gamma(S)$};

\draw [red, very thick] (3.3,-2.8) rectangle (4.7,3.6);
\node[] at (4,3.9) (name) {$\Gamma^2(S)$};

\end{tikzpicture}
\caption{Graph Partitioned into $k-$boundaries}  
\label{period_ex}
\end{figure}
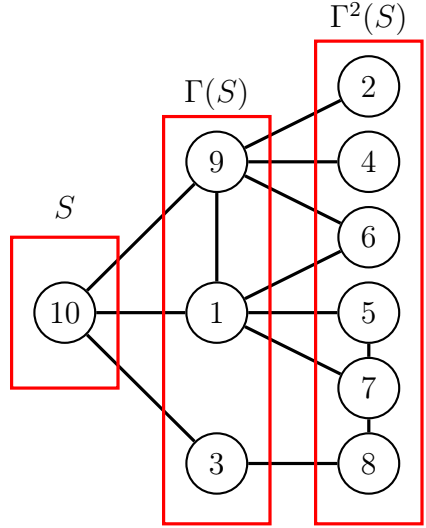

\begin{lemma}\label{con_path}
    Let $(G,f)$ be a graph with valuation $f:V \to [n]$ and consider the Markov chain of possibilities $\mathcal{G}$. Then there exists a directed path in the Markov chain from $f$ to $f_{con}$ where $f_{con}$ is some constant valuation.
\end{lemma}

\begin{proof}
Let $f_0 = f$. Let $$S = \{v \in V | f(v) = M_f\}$$ be the set of vertices with maximum value. Further assign $M_0 = M_f$ as the initial maximum value. $S$ is non-empty so we can consider $k$-boundaries of $S$. 
\par 

To show the existence of such a path in the Markov chain it suffices to provide a sequence of vertices that could be chosen by the asynchronous model to reach a constant valuation. We have from Lemma \ref{nbounpartition} that there exists a $k \in \N \cup \{0\}$ such that
$$V = S \sqcup \Gamma(S) \sqcup \Gamma^2(S) \sqcup \dots \sqcup \Gamma^k(S) = \bigsqcup_{n=0}^k \Gamma(S)^n$$

The path taken through the Markov chain of possibilities will be the following. Select any $v' \in \Gamma(S)$. For all $u \in V$, $f_{0}(u) \leq  M_0 $. Further $f_0(s) = M$ and $s \sim v'$ for $s \in S$. Therefore
updating the model on this node produces 
        \begin{equation*}
            f_1(v) = \begin{cases}
                M_0 & \text{if } v = v' \\
                f_{0}(v) & \text{if } v \neq v' \\
            \end{cases}
        \end{equation*}
The valuation $f_0$ is identical to $f_1$ except now we have one more vertex with value $M_0$, namely a vertex in $\Gamma(S)$. We can continue in this manner, choosing to update all vertices in $\Gamma(S)$, then update all vertices in $\Gamma^2(S)$ and so on until updating all vertices in $\Gamma^k(S)$. The number of updates is finite and equals $|V| - |S|$. Upon completion all vertices will have the value $M_0$. This is a valid sequence of updates, so there exists a path in the Markov chain of possibilities from $(G,f)$ to $(G,f_{con})$.

\end{proof}

The above lemma states that no matter what valuation you start with, it is always possible to trace a path through the Markov chain of possibilities to a graph with a constant valuation. This proof also provides information about the function $f_{con}$. It is a valuation that assigns $M_0$ to every vertex. This leads to a complete characterisation of the size of the absorbing states as we now show.

\begin{theorem}\label{path_in_chain_implies_period_1}
    The period of the asynchronous maximum model for an undirected graph $G$ is 1, that is the size of all absorbing states in the Markov chain of possibilities is 1.
\end{theorem}

\begin{proof}
    Let $\mathcal{G}$ be the Markov chain of possibilities for a graph $G$. Firstly we argue that every absorbing state contains $f_{con}$ where $f_{con}$ is some constant valuation. Then we will show that no absorbing state can have size strictly greater than 1. This implies $f_{con}$ is in its own absorbing state.
    \par 
    Let $\Delta = (\mathcal{V'},\mathcal{E'}) \subseteq \mathcal{G}$ be an absorbing state. Since $\Delta$ is a strongly connected component then there exists a path between any two elements $f'$ and $f$ in $\Delta$. Let $f'$ be in $\Delta$. By Lemma $\ref{con_path}$, there exists a path from $f'$ to $f_{con}$ for some constant valuation $f_{con}$. Since $\Delta$ is absorbing, the path from $f'$ to $f_{con}$ cannot leave $\Delta$. Therefore every absorbing state contains a constant valuation.
    \par Suppose for contradiction that the absorbing state $\Delta$ contains strictly more than one valuation. We have already shown $f_{con} \in \Delta$. Therefore assume there exists $f' \in \Delta$ with $f' \neq f_{con}$. Since $\Delta$ is strongly connected, there exists a path from $f_{con}$ to $f'$. This implies there exists some sequence of vertices that can be updated to obtain $f'$ from $f_{con}$. However this is a contradiction since \textit{any} vertex updated under a constant valuation does not change value. Therefore the absorbing state $\Delta$ has size at most $1$.
    \par
    Since every absorbing state contains the constant valuation and has size at most 1, then every absorbing state must \textit{only} contain a constant valuation. Therefore the period of the asynchronous maximum model is 1.
\end{proof}

\section{Convergence}
Here we analyse the convergence time of the asynchronous maximum model on undirected graphs.

\begin{definition}\label{strong_edge}
    (Strong Edge) Let $(G,f_t)$ be an undirected graph with valuation $f_t$. An edge $(u,v) \in E$ is a \textit{strong edge} if
    $f(u) = f(v) = M_t$. That is, $(u,v)$ is an edge where both $u$ and $v$ contain the current maximum value in $f$.
\end{definition}

\begin{definition}
    (Strong Edge Set) Let $(G,f_t)$ be an undirected graph with valuation $f_t$. The set $S_t\subseteq V$ is a \textit{strong edge set} if all of the vertices in $S_t$ belong to a strong edge and $|S_t|$ is maximal. Note the strong edge set consists of vertices and not edges. If the valuation is clear from context then we write $S = S_t$.
\end{definition}
The strong edge set for a valuation $f$ may be empty if there are no strong edges in the graph. Further, for a given valuation this set is unique, so we are justified in using the terminology `the' strong edge set $S$.

\begin{example}
Please refer to Figure \ref{graphs_ex_val_here}. For the graph on the left, the maximum value in the graph is $5$ and there is a strong edge set $S$ of size 2. Note that not all vertices with the maximum value are in $S$. For the graph on the right, the strong edge set is empty since the maximum value is $6$.
\end{example}

\begin{figure}[h!]
\centering
    \begin{tabularx}{0.8\textwidth}{*{2}{>{\centering\arraybackslash}X}}
\begin{tikzpicture}[
every edge/.style = {draw=black,very thick},
 vrtx/.style args = {#1/#2}{%
      circle, draw, thick, fill=white,
      minimum size=8mm, label=#1:#2}
                    ]
\node(A) [red][vrtx= center/5] at (0, 0) {};
\node(B) [red][vrtx= center/5] at (0, -2) {};
\node(C) [vrtx= center/3] at (2,0) {};
\node(D) [vrtx= center/4] at (2,-2) {};
\node(E) [vrtx= center/5] at (4, 0) {};
\node(F) [vrtx= center/2] at (4, -2) {};
\path   (A) edge (B)
        (A) edge (C)
        (C) edge (D)
        (B) edge (D)
        (C) edge (A)
        (C) edge (E)
        (E) edge (F)
        (D) edge (F)
        (E) edge (D);
\end{tikzpicture}
    \caption*{Valuation with $S$ highlighted in red}  
    &   
\begin{tikzpicture}[
every edge/.style = {draw=black,very thick},
 vrtx/.style args = {#1/#2}{%
      circle, draw, thick, fill=white,
      minimum size=8mm, label=#1:#2}
                    ]
\node(A) [vrtx= center/5] at (0, 0) {};
\node(B) [vrtx= center/5] at (0, -2) {};
\node(C) [vrtx= center/3] at (2,0) {};
\node(D) [vrtx= center/4] at (2,-2) {};
\node(E) [vrtx= center/5] at (4, 0) {};
\node(F) [vrtx= center/6] at (4, -2) {};
\path   (A) edge (B)
        (A) edge (C)
        (C) edge (D)
        (B) edge (D)
        (C) edge (A)
        (C) edge (E)
        (E) edge (F)
        (D) edge (F)
        (E) edge (D);
\end{tikzpicture}
    \caption*{Valuation where $S = \emptyset$}  
    \end{tabularx}
\caption{Two examples of the strong edge set $S$}
\label{graphs_ex_val_here}
\end{figure}

\subsection{Potential Function}
Now we define a potential function on $(G,f_t)$. Let $g : \mathcal{F} \to \mathbb{Z}$ be given by

\begin{align*}
    & g(f_t) = |S_t| && \text{$S_t \subseteq V$ is the strong edge set under valuation $f_t$} \\
\end{align*}

\begin{lemma}\label{potential_func}
    Let $(G,f_0)$ be a graph with initial valuation $f_0$. Update the valuations according to the asynchronous maximum model. The potential function $g_t$ is non-negative, non-decreasing (with respect to $t$) and bounded above by $n$.
\end{lemma}

The above lemma is useful as $g$ can be used as a metric for how close we are to convergence. Note that $g(f_t) = n$ corresponds to $f_t$ being a constant valuation.

\begin{proof}
    $g_t$ is the cardinality of a subset of vertices, so it must be between $0$ and $|V|$. Therefore $g_t$ is non-negative and bounded above by $|V| = n$.
    \par We show $g_{t+1} \geq g_t$ and therefore non-decreasing. Suppose vertex $v' \in V$ is chosen during an update from $t$ to $t+1$. Observe that the only vertex whose value changes is $v'$. Consider the following cases depending on whether or not $v'$ is in the strong edge set.
    \par \textbf{Case 1: $v' \in S_t$}.
    \par If $v' \in S_t$, then by definition it has a neighbour with the current maximum $M_t$ in the graph (see \ref{strong_edge}). The update rule then gives
            \begin{equation*}
            f_t(v) = f_{t+1}(v) = \begin{cases}
                M_t & \text{if } v = v' \\
                f_{t}(v) & \text{if } v \neq v' \\
            \end{cases}
        \end{equation*}
    That is no values on any vertices change. Then $|S_{t+1}| = |S_t|$.
    \par \textbf{Case 2: $v' \notin S_t$}.
    \par If $v' \notin S_t$ then we consider two cases. Whether or not $v'$ has a neighbour with value $M_t$.
    \par If $v'$ has a neighbour $u$ with value $f_t(u) = M_t$ then the update rule gives 
            \begin{equation*}
            f_{t+1}(v) = \begin{cases}
                M_t & \text{if } v = v' \\
                f_{t}(v) & \text{if } v \neq v' \\
            \end{cases}
        \end{equation*}
    Now the edge $(v',u)$ is a strong edge. Hence $|S_{t+1}| \geq |S_t| + 1$ since we have at least included the vertex $v'$ into the strong edge set $S_t$ to obtain $S_{t+1}$. Further, the update of $v'$ does not change the values of any vertices already in $S_t$.
    \par If $v'$ has no neighbour $u$ such that $f_t(u) = M_t$, then we have
            \begin{equation*}
            f_{t+1}(v) = \begin{cases}
                 \text{max}_{v \sim u}\{f_{t}(u)\} & \text{if } v = v' \\
                f_{t}(v) & \text{if } v \neq v' \\
            \end{cases}
        \end{equation*}
By assumption $\text{max}_{v \sim u}\{f_{t}(u)\} < M_{t}$. 
\par If $S_t = \emptyset$ then consider $M_{t+1} \leq M_t$ which is the maximum value in the graph in round $t+1$. If $f_t(u) = M_{t+1}$ then the update will produce a strong edge $(v',u)$ and $|S_{t+1}| \geq 2$. Otherwise if $f_t(u) \neq M_{t+1}$ then the update does not produce a strong edge and $S_{t+1} = \emptyset$.
\par Otherwise if $S_t \neq \emptyset$. By assumption $v'$ has no neighbour $u$ such that $f_t(u) = M_t$. Therefore the update cannot produce a strong edge and $S_t = S_{t+1}$.
\end{proof}

We note that in the above proof we show that $S_t \subseteq S_{t+1}$ for all $t$. This is slightly stronger that proving the potential function is non-decreasing. It shows once an element is added to the strong edge set it cannot be removed by any update. The advantage of a potential function is that we can study how it increases with respect to the update process. In fact, an increase in the potential function corresponds to transitioning between maximal strongly connected components of the Markov chain of possibilities $\mathcal{G}$. Since $g_t(f) \leq n$ then it can only increase a finite number of times which will allow us to bound the expected time taken to reach an absorbing state.

\subsection{Bounds on Convergence Time}

Now we can try to analyse the convergence time of the asynchronous maximum model. In section \ref{period_sec}, the absorbing states in the Markov chain of possibilities were shown to be constant valuation states. These exactly correspond to strong edge sets of size $|V|$. Lemma $\ref{potential_func}$ prompts us to study the following two stages of convergence:
\begin{enumerate}
    \item How many rounds until a strong edge is formed?
    \item Once a strong edge exists, how many rounds until the strong edge set $S \subseteq V$ is the entire vertex set $V$?  \\
\end{enumerate}

We will call these Phase 1 and Phase 2 respectively.

\par

\begin{theorem}\label{asymptotics_G}
    For any graph $G$ with $n$ vertices,
    \begin{equation}
        \Omega(n\log n) = \T(G) = \bigO(n^2)
    \end{equation}
\end{theorem}

\begin{proof}
    Firstly we show $\T(G) = \Omega(n\log n)$. Let $f'$ be some valuation that assigns the value $2$ to two adjacent vertices and the value $1$ to every other vertex. Since there exists a strong edge between the two adjacent values of 2 which are maximal then the absorbing state will be reached when all values in the graph have been updated to 2.  Let $Z$ denote a random variable that is the number of rounds until $S$ is equal to $V$ for valuation $f'$. Let $Z_i$ be a random variable denote the number of rounds until $|S| = i$ given that $|S| = i-1$. We can write $Z = Z_3 + Z_4 + \dots + Z_n$ since it must be the case that $|S|$ increases by $0$ or $1$ after an update for this valuation $f'$ (this is not true for an arbitrary valuation $f$).
    Then
    $$ \sum_{i=3}^n \mathbb{E}[Z_i] = \mathbb{E}[Z] = T(G,f') \leq \max_{f \in [\mathcal F]}T(G,f) = \T(G).$$    
    Let $|S|$ denote the current size of the strong edge set. We consider an update successful if it increases the size of the strong edge set. The probability of success is given by $\frac{|\Gamma(S)|}{n} \leq \frac{n-|S|}{n}$. This is because $|\Gamma(S)|$ denotes the neighbours of the set $S$ not already in $S$ and of any one of these vertices is chosen then the size of $S$ increases. We can consider each $Z_i$ as being \textit{stochastically dominated} by a geometric random variable with success probability $\frac{n-|S|}{n} = \frac{n-(i-1)}{n}$. Therefore the expectation of this geometric random variable lower bounds $\mathbb{E}[Z_i]$. Then we have
    $$\frac{n}{n-(i-1)} \leq \mathbb{E}[Z_i].$$
    Therefore
    $$\sum_{i=3}^n \frac{n}{n-{i-1}} \leq \sum_{i=3}^n \mathbb{E}[Z_i] = \E[Z] $$
    We can now evaluate the sum.
    \begin{align*}
        & \sum_{i=3}^n \frac{n}{n-{i-1}}  = n \sum_{i=3}^n \frac{1}{n-{i-1}} && \\ 
        & = n \sum_{j=1}^{n-2} \frac{1}{j} && \\ 
         & = n H_{n-2} && \text{} \\ 
         & = \Omega(n\log n) && \text{Lemma \ref{harmonic_num}} \\ 
    \end{align*}
    Where $H_n$ is the $n$-th harmonic number. \newline

    \par Now we show $\T(G) =\bigO(n^2)$. Let $f$ be any valuation of $G$. In the worst case, the update process would need to first produce a strong edge. Further, in the worst case the size of the strong edge set would increase by at most $0$ or $1$ when the correct vertex is chosen. We will use these two assumptions on the worst case to show the result.
    \par Let $X$ be the number of rounds until a strong edge is formed. For any given round $i$ the probability $p_i$ that a strong edge forms is lower bounded by $\frac{n}{\delta}$ where $\delta \geq 1$ is the minimum degree of the graph $G$. If we consider $X'$ a geometric random variable with success probability $\frac{\delta}{n}$, then $X'$ is stochastically bounded by $X$. Therefore
    $$\E[X] \leq \E[X'] = \frac{n}{\delta}.$$
    \par Let $Z$ be the number of rounds until $S = V$ given that there exists exactly one strong edge in the graph $G$. In the worst case, a successful update would only increase the size of $S$ by $1$ (there exist graphs and valuations $|S|$ can change from $|S| = 0$ to $|S| = n$ in one round. Let $Z_i$ denote the number of rounds until $|S| = i$ given that $|S| = i-1$. Then $Z = \sum_{i=3}^n Z_i$. Further, the probability of success is at least $\frac{|\Gamma(S)|}{n} \geq \frac{1}{n}$ since the graph is connected. Therefore each $Z_i$ stochastically dominates a random variable $Z_i'$ which is geometrically distributed with success probability $\frac{1}{n}$ and
    $$\E[Z] = \sum_{i=3}^n \E[Z_i] \leq \sum_{i=3}^n \E[Z_i']  =\sum_{i=3}^n n = n(n-2) $$
    Then letting $Y$ be a random variable denoting the number of rounds until $S = V$, we have that
    $$\E[Y] = \E[X + Z] = \E[X] + \E[Z] \leq \frac{n}{\delta} + n(n-2) \leq n + n^2 - 2n \leq n(n-1) \leq n^2$$
    Therefore $\E[Y] = \bigO(n^2)$ and since $f$ was an arbitrary valuation we have $\T(G) = \bigO(n^2)$.
\end{proof}

We have gained the following insights. Phase 1 has a dependence on the minimum degree $\delta$. Phase 2 seems to be the dominating factor of convergence in since Phase 1 is expected to finish in $n$ rounds in the worst case while Phase 2 is expected to take at least $n\log n$ rounds. Further, Phase 2 depends on $|\Gamma(S)|$ where $S$ is the strong edge set.

\subsection{Tightness of Bounds}
We show there exist graphs that obtain the bounds $\bigO(n \log n)$ and $\Omega(n^2)$.

\begin{lemma}\label{complete_graph_time}
    The complete graph $K_n$ has convergence time $\T(K_n) = \Theta(n \log n)$.
\end{lemma}
The proof is essentially the same idea as used in Theorem \ref{asymptotics_G}.
\begin{proof}
    $\T(K_n) = \Omega(n \log n)$ by Theorem $\ref{asymptotics_G}$. Then it suffices to show $\T(K_n) = \max_{f \in [\mathcal{F}]} T(K_n,f) = \bigO(n \log n)$. 
    \par Let $f \in [\mathcal{F}]$ be any valuation of $K_n$. Again we analyse Phase 1 and Phase 2. Let $X$ denote a random variable that is the number of rounds until a strong edge is formed. We have that either $\mathbb{P}(X = 0) = 1$ or $\mathbb{P}(X = 1) = 1$. The former occurs when a strong edge already exists in the valuation $f$. The latter is because we are guaranteed to either choose a neighbour of the current maximum vertex or otherwise we choose the current maximum vertex and both of these result in a strong edge. Therefore $\E[X] \leq 1$. 
    \par Now we analyse phase 2. Let $Z$ be the number of rounds until we reach an absorbing state. The size of the strong edge set $S$ can increase by at most $1$ in each round. This is because for a complete graph, all vertices with the maximum value are in $S$ and we only update one vertex per round. Let $Z_i$ be the random variable denoting the number of rounds until $|S| = i$ given $|S| = i-1$ we have that $Z = \sum_{i=3}^{n} Z_i$. Since $K_n$ is the complete graph we have that the probability of a successful update is exactly $\frac{n-|S|}{n}$. So each $Z_i$ is a geometric random variable with success probability $\frac{n-|S|}{n} = \frac{n-(i-1)}{n}$. The expectation is then
    $$\E[Z] = \sum_{i=3}^n \E[Z_i] = \sum_{i=3}^n \frac{n}{n-(i-1)} = n\sum_{j = 1}^{n-2} \frac{1}{j} = nH_{n-2} = \bigO(n \log n)$$
    The last equality is by Lemma $\ref{harmonic_num}$. Letting $Y = X + Z$ be the number of rounds until we reach a constant valuation state we have
    $$T(K_n,f) = \E[Y] = \E[X] + \E[Z] \leq 1 + \E[Z] = 1 + \bigO(n \log n) = \bigO(n \log n)$$
\end{proof}

\begin{lemma}\label{path_graph_time}
    There exists a valuation $f$ of the path graph $P_n$ such that $T(P_n,f) = \Omega(n^2)$. Therefore $\T(P_n) = \Theta(n^2)$
\end{lemma}
\begin{proof}
    $\T(P_n) \in \bigO(n^2)$ by Theorem \ref{asymptotics_G} so it suffices to show $\T(P_n) \in \Omega(n^2)$. Let $f'$ be the valuation given in Figure \ref{fig:path_val} that assigns the value $2$ to an endpoint and its neighbour and the value $1$ to every other vertex.
    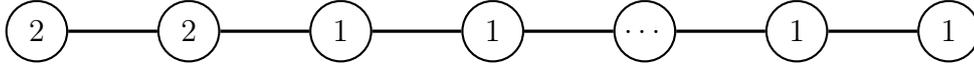
\begin{figure}[h!]
    \centering
    \begin{tikzpicture}[
every edge/.style = {draw=black,very thick},
 vrtx/.style args = {#1/#2}{%
      circle, draw, thick, fill=white,
      minimum size=8mm, label=#1:#2}
                    ]
\node(A) [vrtx= center/2] at (0, 0) {};
\node(B) [vrtx= center/2] at (2, 0) {};
\node(C) [vrtx= center/1] at (4, 0) {};
\node(D) [vrtx= center/1] at (6, 0) {};
\node(E) [vrtx= center/\dots] at (8, 0) {};
\node(F) [vrtx= center/1] at (10, 0) {};
\node(G) [vrtx= center/1] at (12, 0) {};

\path   (A) edge (B)
        (B) edge (C)
        (C) edge (D)
        (D) edge (E)
        (E) edge (F)
        (F) edge (G);
\end{tikzpicture}
    \caption{Valuation of $P_n$ achieving $\Omega(n^2)$ convergence time}
    \label{fig:path_val}
\end{figure}

    Then $T(P_n,f') \leq \T(P_n) = \max_{f \in [\mathcal{F}]} T(P_n,f)$. This valuation already has a strong edge so we are in phase 2. Let $Z$ be the number of rounds until $S = V$. Let $Z_i$ denote the number of rounds until $|S| = i$ given $|S| = i-1$. Then $Z_i$ is a geometric random variable with success probability $ \frac{\Gamma(S)}{n}=\frac{1}{n}$ since $\Gamma(S) = 1$ for $2 \leq |S| < n$. Then
    $$T(P_n,f') = \E[Z] = \sum_{i=3}^n \E[Z_i] = \sum_{i=3}^n n = n(n-2).$$
    Note that for $n\geq 4$, $\frac{1}{2}n^2\leq n(n-2) = T(P_n,f') \leq \T(P_n)$. Therefore $\T(P_n) = \Omega(n^2)$.
    
\end{proof}

\subsection{Bounds in Terms of Vertex Expansion}

It has already been remarked that Phase 2 seems to dominate the convergence time and this depends on the value of $|\Gamma(S)|$ as $|S|$ increases. We note that $K_n$ which converges quickly under this model is a highly connected graph while $P_n$ which converges slowly is not well connected. These insights prompt us to study the convergence time in terms of the vertex expansion $\phi$ given in Definition \ref{vertex_exp_def_undirected}.
\par

\begin{lemma}\label{vertex_expansion_bound}
    Let $G$ be a strongly connected graph on $n \geq 2$ vertices. Then
    $$0<\frac{2}{n} \leq \phi_{out}(G) \leq 5$$
    So $\phi_{out}(G)$ is bounded by constants.
\end{lemma}
\begin{proof}
    Please see the Appendix at \ref{vertex_bound_proof}.
\end{proof}

\begin{theorem}\label{asymptotics_vertex}
    The convergence time of $G$ with vertex expansion $\phi$ is in $\bigO(\frac{n}{\phi}\log n)$.
\end{theorem}
The above theorem characterises an upper bound on $\T(G)$ in terms of the property $\phi$. This is interesting as now we potentially have a better upper bound than $\bigO(n^2)$ for a large class of graphs. In particular, graphs with nearly constant vertex expansion $\phi$ should converge quickly.

\begin{proof}
    Let $f$ be the initial valuation of $G$ that achieves the maximum expected number of rounds. So $T(G,f) = \T(G).$ We show that $T(G,f) \in \bigO(\frac{n}{\phi}\log n)$
    We consider three sections in this proof depending on the size of the strong edge set $|S|$. In each section we bound the expected number of rounds for the strong edge set to be in the specified size interval.
    \par \textbf{Section 1:} $0 \leq |S| < 2.$
    \par
    Without loss of generality we can assume the valuation $f$ has no strong edges so $|S| = 0$. Let $X$ be a random variable denoting the number of rounds until a strong edge is formed. For any given round, the probability that a strong edge is formed is at least $\frac{\delta}{n}$. Therefore $X$ stochastically dominates $X'$ where $X'$ is a geometric random variable with success probability $p = \frac{\delta}{n}$. Therefore $\E[X] \leq \E[X'] = \frac{n}{\delta}$.
    
    \par \textbf{Section 2:} $2 \leq |S| \leq \lfloor \frac{n}{2} \rfloor$. After a strong edge has formed we can leverage the definition of vertex expansion up until $|S| \leq \frac{n}{2}$. Let $Z$ be a random variable denoting the number of rounds until $|S| > \frac{n}{2}$ given $|S| = 2$. Let $Z_i$ be a random variable denoting the number of rounds until $|S| = i$ given that $|S| = i-1$. In the worst case we assume that for $|S|\geq 2$, $|S|$ increases by at most one on a successful update.
    \par Assume $|S| = i-1$ where $2 \leq i-1 \leq \lfloor \frac{n}{2}\rfloor$. We call an update successful if it increases $|S|$. For any given round the probability of a successful update is $\frac{\Gamma(S)}{n}$. However by definition of $\phi$ we have that $\phi|S| \leq \Gamma(S)$. Therefore $Z_i$ stochastically dominates a geometric random variable $Z_i'$ which has success probability $\frac{\phi|S|}{n} = \frac{\phi(i-1)}{n}$. Then $\E[Z_i] \leq \E[Z_i'] = \frac{n}{\phi(i-1)}$ and we have
    $$\E[Z] = \sum_{i=3}^{\lfloor \frac{n}{2} \rfloor+1} \E[Z_i] \leq \frac{n}{\phi} \sum_{i=3}^{{\lfloor \frac{n}{2} \rfloor}+1}  \frac{1}{i-1} = \frac{n}{\phi} \sum_{i=2}^{\lfloor \frac{n}{2} \rfloor} \frac{1}{i} = \frac{n}{\phi}(H_{\lfloor \frac{n}{2} \rfloor} - 1)= \bigO(\frac{n}{\phi}\log n).$$
    So $\E[Z] = \bigO(\frac{n}{\phi}\log n)$.
    \par \textbf{Section 3:} $\lfloor \frac{n}{2} \rfloor + 1 \leq |S| < n$. Please refer to Figure $\ref{fig:phase3}$. We partition $G$ into the sets $S, W = \Gamma(S)$ and $U = V \setminus (S \cup \Gamma(S))$.

    \begin{figure}[h!]
        \centering
        \includegraphics[scale = 0.4]{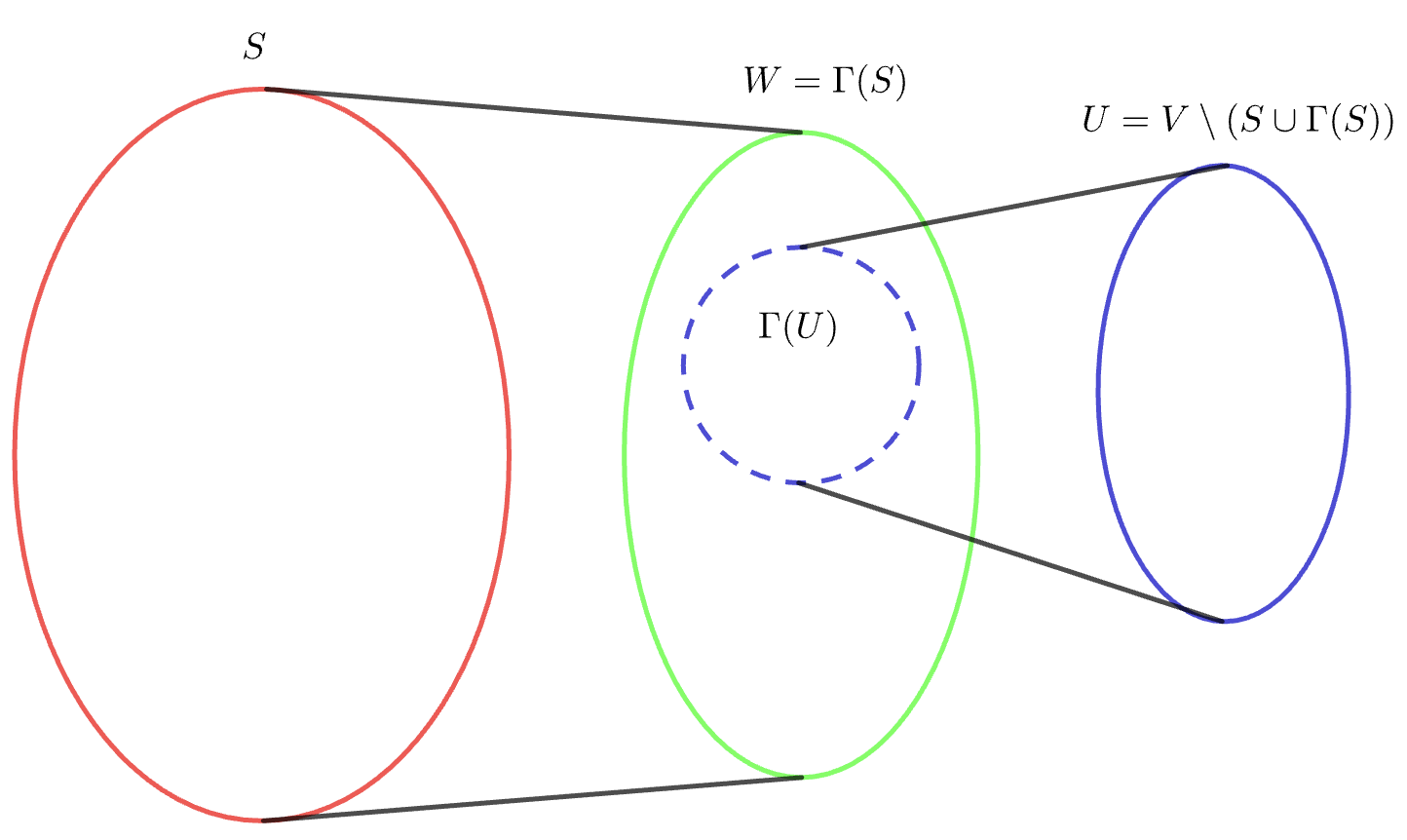}
        \caption{Section 3: Vertices of $G$ partitioned into sets}
        \label{fig:phase3}
    \end{figure}

In the worst case an update will be successful only when a vertex $v \in W$ is selected. The set $\Gamma(U)$ is easier to control than $W$. in terms of the parameter $\phi$. We first consider the number of rounds to first shrink $|U|$ to $0$ and then the number of rounds to shrink $|W|$ to $0$ (both in the worst case).

\par Let $Y$ be a random variable denoting the number of rounds until $|U| = 0$. Let $Y_i$ be the number of rounds until $|U| = i$ given $|U| = i + 1$. We make the assumption that in the worst case $|U|$ can decrease by at most one during an update. Let $|U| = m \leq \lfloor \frac{n}{2} \rfloor -2 $ be the initial size of $U$. Then $Y = \sum_{i = 0}^{m} Y_i$. In order to decrease $|U|$ by one, we must choose a vertex in $\Gamma(U)$. This occurs with probability $\frac{\Gamma(U)}{n}$. Further, $\phi|U| \leq \Gamma(U)$. This implies $\frac{\Gamma(U)}{n} \geq \frac{\phi |U|}{n}$. Therefore each $Y_i$ stochastically dominates a geometric random variable $Y_i'$ with success probability $\frac{\phi |U|}{n} = \frac{\phi(i+1)}{n}$.
We have that

$$\E[Y] = \sum_{i = 0}^{m-1} \E[Y_i] \leq \sum_{i = 0}^{m-1} \frac{n}{\phi(i+1)}$$
Now using that $m \leq \lfloor \frac{n}{2} \rfloor -2  $ gives

$$\leq \sum_{i = 0}^{\lfloor \frac{n}{2} \rfloor -3} \frac{n}{\phi(i+1)} = \frac{n}{\phi} \sum_{i = 0}^{\lfloor \frac{n}{2} \rfloor -3} \frac{1}{i+1} \leq \frac{n}{\phi}H_{\lfloor \frac{n}{2} \rfloor } = \bigO(\frac{n}{\phi}\log n).$$
So $\E[Y] = \bigO(\frac{n}{\phi}\log n)$.

\par Now we assume $|U| = 0$ and consider the number of rounds to shrink $|W| = |\Gamma(S)|$ to $0$. Let $|W| = p \leq \lfloor \frac{n}{2} \rfloor - 1$ be the initial size of $W$. Since $|U| = 0$ then the vertices of $G$ are either in $S$ or in $W$. Let $R$ be a random variable denoting the number of rounds until $|W| = 0$ given $|W| = p$. Let $R_i$ denote the number of rounds until $|W| = i$ given that $|W| = i+1$. Then $R = \sum_{i = 0}^{p-1} R_i$. Each $R_i$ is a geometric random variable with success probability $\frac{|W|}{n} = \frac{i+1}{n}$. Then the expected number of rounds until a successful update is $\frac{n}{i+1}$. Therefore
$$\E[R] = \sum_{i = 0}^{p-1} \E[R_i] \leq \sum_{i = 0}^{p-1} \frac{n}{i+1}  $$
Now using that $p \leq \lfloor \frac{n}{2} \rfloor - 1$ gives
$$ \leq \sum_{i = 0}^{\lfloor \frac{n}{2} \rfloor - 2} \frac{n}{i+1} \leq n \sum_{j=1}^n \frac{1}{i} = n H_n = \bigO(n\log n) = \bigO(\frac{n}{\phi}\log n). $$

\par The last equality is by Lemma \ref{vertex_expansion_bound} since we can bound $0 < \phi \leq 5$. The expected number of rounds until $|S| = n$ is then

$$\E[X] + \E[Z] + \E[Y] + \E[R] = \frac{n}{\delta} + O(\frac{n}{\phi}\log n) + O(\frac{n}{\phi}\log n) + O(\frac{n}{\phi}\log n) = O(\frac{n}{\phi}\log n).$$

Then we have shown $\T(G) = O(\frac{n}{\phi} \log n)$.
    
\end{proof}

We remark that $K_n$ has constant vertex expansion, so $\T(K_n) = \bigO(n\log n)$ agrees with the bound in Lemma \ref{complete_graph_time}. However the path graph $P_n$ has $\phi \approx \frac{2}{n}$ which gives $\T(P_n) = \bigO(n^2\log n)$. This does not agree with Theorem $\ref{asymptotics_G}$ which implies $\T(P_n) = O(n^2)$. In Section 2 (and 3) of Theorem $\ref{asymptotics_vertex}$ using the term $\phi|S|$ to lower bound $|\Gamma(S)|$ is quite `loose' and this may account for the $\log n$ factor which is gained.

\begin{question}
    Can the bound of $\bigO(\frac{n}{\phi}\log n)$ in Theorem \ref{asymptotics_vertex} be improved as to remove the `extra' $\log n$ factor that appears for certain graphs i.e $P_n$? Potentially a different property to the vertex expansion would need to be used to characterise $\T(G)$. It is interesting to note that a bound of the form $\bigO(\frac{n}{\phi}\log (\phi n))$ would be tight for both the path graph $P_n$ ($\phi(P_n)  \approx \frac{2}{n}$) and the complete graph $K_n$ ($\phi(K_n) \approx 1$). Is this the correct formula and can it be derived?
\end{question}

\subsection{Concentration of Convergence}

So far we have analysed the expectation of the random variable $A$ denoting the number of rounds until $S = V$. However it is beneficial to show there is a small probability that the number of rounds is much larger than $\E[A]$. This amounts to showing that $A$ is `concentrated' around its expectation.

\begin{definition} \cite{whp_cite}
    Let $A_n$ denote some event that depends on $n \in \N$ (for example graphs on $n$ vertices with some property). Let $\mathbb{P}(A_n)$ denote the probability of $A_n$ . We say $A_n$ occurs \textit{with high probability} when
    $$\mathbb{P}(A_n) \geq 1 - n^{-\Omega(1)}$$
    as $n \to \infty$.
\end{definition}

An event occurring with high probability means as $n \to \infty$ its probability tends to $1$.

\begin{lemma}\label{var_lemma}
    Suppose $X'$ and $X$ are random variables and $X'$ is stochastically dominated by $X$. Then
    $$Var[X] \leq Var[X'] + \E[X']^2.$$
\end{lemma}
\begin{proof}
    Since $X'$ is stochastically dominated by $X$ and the map $x \to x^2$ is non-decreasing, then $\E[X^2] \leq \E[X'^2]$. This gives $Var[X] = \E[X^2] - \E[X]^2 \leq \E[X'^2]$. Therefore
    $$Var[X] - \E[X']^2 \leq \E[X'^2]-\E[X']^2  = Var[X']$$
    which implies the result.
\end{proof}

\begin{theorem}\label{conv_time_whp}
    Let $(G,f)$ be a graph and let $A$ be a random variable denoting the number of rounds until a constant valuation state is reached. Let $a = (\frac{n}{\phi}\log n) g(n)$ where $g(n)$ is any function such that $\frac{1}{g^2(n)} \to 0$ as $n \to \infty$. 
    \par Then the event $\{|A - \E[A]| \leq a\}$ occurs with high probability. That is the number of rounds until convergence is in $\bigO((\frac{n}{\phi}\log n)g(n))$ with high probability.
\end{theorem}

The above theorem tells us that for large $n$, if we choose a random initial valuation $f_0$, it is almost guaranteed to take $\bigO((\frac{n}{\phi}\log n)g(n))$ rounds to converge.

\begin{proof}
    The idea will be to apply Chebyshev's inequality to the random variable $A$. We first stochastically dominate $A$ by $A'$ which is a sum of independent random variables and then apply Lemma \ref{var_lemma}. Then we only need to find the asymptotic behaviour of $Var[A'] + \E[A']^2$.
    \par 
    Let the random variables $X,Z,Y,R$ be given as in Theorem \ref{asymptotics_vertex}. Further let $S,U$ and $W$ be subsets of the vertices as given in the theorem. Recall that the random variables denote the following stages of convergence.
    \begin{itemize}
        \item     $X:$ Number of rounds until a strong edge
        \item     $Z:$ Number of rounds where $2 \leq |S| \leq \lfloor \frac{n}{2} \rfloor$.
        \item    $Y:$ Number of rounds for the set $U$ to shrink to $|U| = 0$.
        \item     $R:$ Number of rounds for the set $W$ to shrink to $|W| = 0$.
    \end{itemize}
    Note that $A = X+Z+Y+R$ is the random variable denoting the number of rounds until convergence. This is stochastically dominated by the random variable $A' = X' + \sum_{i=3}^{\lfloor \frac{n}{2} \rfloor }Z_i' + \sum_{j = 0}^{\lfloor \frac{n}{2} \rfloor - 3} Y_j' + \sum_{k = 0}^{\lfloor \frac{n}{2} \rfloor - 2} R_i$ where $X'$, $Z_i', Y_j'$ and $R_i$ are as given in Theorem $\ref{asymptotics_vertex}$. We do this as now $A'$ is a sum of independent random variables. By Lemma \ref{var_lemma} we then have
    \begin{equation}\label{var_a_bound}
        Var[A] \leq Var[A'] + \E[A']^2 = Var[A'] + O(\frac{n^2}{\phi^2}\log^2 n).
    \end{equation}
    Now we want to asymptotically bound 
    \begin{equation}\label{var_sum}
        Var[A'] = Var[X'] + \sum_{i=3}^{\lfloor \frac{n}{2} \rfloor }Var[Z_i'] + \sum_{j = 0}^{\lfloor \frac{n}{2} \rfloor - 3} Var[Y_j'] + \sum_{k = 0}^{\lfloor \frac{n}{2} \rfloor - 2} Var[R_k]
    \end{equation}
    Each of these random variables are geometrically distributed with some success probability $p$ which varies. The variance is given by $\frac{1-p}{p^2}$. We will bound each term in Equation \ref{var_sum} individually.
    \par $X'$ is geometrically distributed with success probability $\frac{\delta}{n}$. Then
    $$Var[X'] = \frac{1-\frac{\delta}{n}}{\frac{\delta^2}{n^2}} = \frac{n(n-\delta)}{\delta^2} = O(n^2) = O(\frac{n^2}{\phi^2}).$$
    \par Each $Z_i'$ is a geometric random variable with success probability $\frac{\phi(i-1)}{n}$. Therefore
    \begin{align*}
        & \sum_{i=3}^{\lfloor \frac{n}{2} \rfloor} Var[Z_i'] = \frac{n}{\phi^2}\sum_{i=3}^{\lfloor \frac{n}{2} \rfloor} \frac{n-\phi(i-1)}{(i-1)^2}  && \text{} \\
        &  = \frac{n}{\phi^2}\sum_{j=3}^{\lfloor \frac{n}{2} \rfloor-1} \frac{n-\phi j}{j^2}  && \text{} \\
        &  = \frac{n}{\phi^2}\left(n\sum_{j=3}^{\lfloor \frac{n}{2} \rfloor-1} \frac{1}{j^2} - \phi \sum_{j=3}^{\lfloor \frac{n}{2} \rfloor-1} \frac{1}{j}  \right) && \text{} \\
        &  \leq \frac{n}{\phi^2}\left(n\sum_{j=3}^{\lfloor \frac{n}{2} \rfloor-1} \frac{1}{j^2}  \right) && \text{} \\
        &  \leq \frac{n^2}{\phi^2} \cdot \frac{\pi^2}{6}&& \text{} \\
        &  = \bigO(\frac{n^2}{\phi^2})&& \text{} \\
    \end{align*}

\par Each $Y_j'$ is a geometric random variable with success probability $\frac{\phi(j+1)}{n}$. The calculation is almost identical to the previous one to get $\sum_{j=0}^{\lfloor \frac{n}{2} \rfloor-3} Var[Y_j'] = O(\frac{n^2}{\phi^2})$.
\par Each $R_k$ is a geometric random variable with success probability $\frac{k+1}{n}$. Therefore
    \begin{align*}
        & \sum_{k=0}^{\lfloor \frac{n}{2} \rfloor-2} Var[R_k'] = n \sum_{k=0}^{\lfloor \frac{n}{2} \rfloor-2} \frac{n-\phi(k+1)}{(k+1)^2}  && \text{} \\
        &  = n \left(n\sum_{k=0}^{\lfloor \frac{n}{2} \rfloor-2} \frac{1}{(k+1)^2} -  \sum_{k=0}^{\lfloor \frac{n}{2} \rfloor-2} \frac{1}{k+1}  \right) && \text{} \\
        &  \leq  n \left(n\sum_{k=0}^{\lfloor \frac{n}{2} \rfloor-2} \frac{1}{(k+1)^2}\right) && \text{} \\
        &  \leq  n^2 \cdot \frac{\pi^2}{6} && \text{} \\
        &  =  \bigO(n^2) && \text{} \\
        &  =  \bigO(\frac{n^2}{\phi^2}) && \text{} \\
    \end{align*}

Combining these results to Equation \ref{var_a_bound} gives
$$ Var(A) \leq \bigO(\frac{n^2}{\phi^2}) + \bigO(\frac{n^2}{\phi^2}\log^2 n) = \bigO(\frac{n^2}{\phi^2}\log^2 n).$$
The purpose of bounding the variance is to apply Chebyshev's Inequality. For $a > 0$ it gives us that
$$\mathbb{P}(|A - \E[A]| \geq a) \leq \frac{Var[A]}{a^2} \leq \frac{\bigO(\frac{n^2}{\phi^2}\log^2 n)}{a^2} \leq \frac{C\frac{n^2}{\phi^2}\log^2 n}{a^2}$$
for some constant $C$ when $n$ is large. Choose $a = (\frac{n}{\phi}\log n)g(n)$ where $g(n)$ is such that $\frac{1}{g^2(n)} \to 0$ as $n \to \infty$. Then
$$\mathbb{P}(|A - \E[A]| \geq a) \leq \frac{C\frac{n^2}{\phi^2}\log^2 n}{a^2} \leq \frac{C}{g^2(n)}.$$
Therefore $\mathbb{P}(|A - \E[A]| \geq a)  \to 0$ as $n\to \infty$.
\par 
Since $\E[A] = \bigO(\frac{n}{\phi}\log n)$ and the RHS $\to 0$ as $n \to \infty$ the above statement gives us that with high probability the number of rounds until convergence is $\E[A] + a =\bigO(\frac{n}{\phi}\log n) + (\frac{n}{\phi}\log n)(g(n))$. Therefore with high probability the number of rounds until convergence is
\begin{equation}\label{conv_time_whp_eq}
    \bigO((\frac{n}{\phi}\log n)g(n))
\end{equation}

\end{proof}

%% file: chapters/directed.tex
Here we characterise the period and convergence time for strongly connected graphs. We show that the period of the process is $1$. Since the set of all connected undirected graphs are a subset of all strongly connected directed graphs then this result will imply Theorem \ref{path_in_chain_implies_period_1}. The convergence time of the process is shown to be in $\bigO(nb^2 + \frac{n}{\phi'}\log n)$ where $b$ and $\phi'$ are graph parameters defined in Definition \ref{orbit_def} and Theorem \ref{strong_cycle_set_bound}.

\section{Period}
The period of the asynchronous maximum model on a strongly connected graph is 1. The analysis is similar to the undirected graph case.

\begin{lemma}\label{path_in_MC_SCC}
    Let $(G,f)$ be a graph with valuation $f : V \to [n]$ and consider the Markov chain of possibilities $\mathcal{G}$. There exists a directed path in $\mathcal{G}$ from $f$ to $f_{con}$ where $f_{con}$ is some constant valuation.
\end{lemma}
We will sketch the proof idea with an example. Please refer to Figure \ref{sc_example_period}. To prove the result we show there exists a sequence of vertices that could be chosen by the model to reach a constant valuation. Let $S$ be the set of vertices with the current maximum value in the graph $G$. We then partition the graph $\widetilde{G}$ (which is $G$ but with the edge directions reversed) into $k$-boundaries from $S$. The vertices \textbf{from} $G$ we choose are the following. One at a time choose a vertex from $\Gamma_{\widetilde{G}}(S)$. Since there is an edge (in $G$) from $u \in \Gamma_{\widetilde{G}}(S)$ to $v \in S$ then the value at $u$ is updated to the maximum value $10$. After every vertex in $\Gamma_{\widetilde{G}}(S)$ has been chosen we choose, one at a time, every vertex in $\Gamma_{\widetilde{G}}^2(S)$. This will update these vertices to the value $10$. After every vertex in $\Gamma_{\widetilde{G}}^2(S)$ has been selected we have that every value in the graph is $10$. We now formalise this in the proof below.

\begin{figure}[h!]
\centering
    \begin{tabularx}{0.8\textwidth}{*{2}{>{\centering\arraybackslash}X}}
\begin{tikzpicture}[
every edge/.style = {draw=black,very thick},
 vrtx/.style args = {#1/#2}{%
      circle, draw, thick, fill=white,
      minimum size=8mm, label=#1:#2}]
\node(A) [vrtx= center/10] at (0, 0) {};
\node(B) [vrtx= center/9] at (2, 2.5) {};
\node(C) [vrtx= center/1] at (2, 0) {};
\node(D) [vrtx= center/3] at (2, -3) {};
\node(E) [vrtx= center/2] at (4, 4.5) {};
\node(F) [vrtx= center/4] at (4, 3) {};
\node(G) [vrtx= center/6] at (4, 1.5) {};
\node(H) [vrtx= center/5] at (4, 0) {};
\node(I) [vrtx= center/7] at (4, -1.5) {};
\node(J) [vrtx= center/8] at (4, -3) {};

\path   (B) edge[->] (A)
        (C) edge[->] (B)
        (C) edge[->] (A)
        (D) edge[->] (A)
        (E) edge[->] (B)
        (F) edge[->] (B)
        (G) edge[->] (B)
        (G) edge[->] (C)
        (H) edge[->] (C)
        (I) edge[->] (H)
        (J) edge[->] (I)
        (I) edge[->] (C)
        (J) edge[->] (D)
        (A) edge[->] (J)
        (H) edge[->] (G)
        (G) edge[->] (F)
        (F) edge[->] (E);

\end{tikzpicture}
    \caption*{Graph $G$}  
    &   
\begin{tikzpicture}[
every edge/.style = {draw=black,very thick},
 vrtx/.style args = {#1/#2}{%
      circle, draw, thick, fill=white,
      minimum size=8mm, label=#1:#2}]
\node(A) [vrtx= center/10] at (0, 0) {};
\node(B) [vrtx= center/9] at (2, 2.5) {};
\node(C) [vrtx= center/1] at (2, 0) {};
\node(D) [vrtx= center/3] at (2, -3) {};
\node(E) [vrtx= center/2] at (4, 4.5) {};
\node(F) [vrtx= center/4] at (4, 3) {};
\node(G) [vrtx= center/6] at (4, 1.5) {};
\node(H) [vrtx= center/5] at (4, 0) {};
\node(I) [vrtx= center/7] at (4, -1.5) {};
\node(J) [vrtx= center/8] at (4, -3) {};

\path   (B) edge[<-] (A)
        (C) edge[<-] (B)
        (C) edge[<-] (A)
        (D) edge[<-] (A)
        (E) edge[<-] (B)
        (F) edge[<-] (B)
        (G) edge[<-] (B)
        (G) edge[<-] (C)
        (H) edge[<-] (C)
        (I) edge[<-] (H)
        (J) edge[<-] (I)
        (I) edge[<-] (C)
        (J) edge[<-] (D)
        (A) edge[<-] (J)
        (H) edge[<-] (G)
        (G) edge[<-] (F)
        (F) edge[<-] (E);

\draw [red, very thick] (-0.7,-1) rectangle (0.7,1);
\node[] at (0,1.4) (name) {$S$};

\draw [red, very thick] (1.3,-3.7) rectangle (2.7,3.3);
\node[] at (2,3.7) (name) {$\Gamma_{\widetilde{G}}(S)$};

\draw [red, very thick] (3.3,-3.7) rectangle (4.7,5.2);
\node[] at (4,5.6) (name) {$\Gamma_{\widetilde{G}}^2(S)$};

\end{tikzpicture}
\caption*{Graph $\widetilde{G}$ partitioned into $k$-boundaries}  
    \end{tabularx}
\caption{$G$ and the dual (reversed edge graph) $\widetilde{G}$}
\label{sc_example_period}
\end{figure}

\begin{proof}
To show the existence of such a path it suffices to provide a sequence of vertices that could be chosen by the asynchronous model to reach a constant valuation. If $f$ is already a constant valuation we are done so we assume it is not.

    Let $f_0 = f$ and let
    $$S = \{v \in V | f_0(v) = M_{f}\}$$ 
    be the set of vertices with maximum value. Further assign $M_0 = M_f$ as the initial maximum value.
\par 
    $S$ is non-empty and $G$ is strongly connected, therefore $\widetilde{G}$ is strongly connected. We now partition the vertices of $\widetilde{G}$ according to $k$-boundaries of $S$. From Lemma \ref{nbounpartition} there exists a $k \in \N \cup \{0\}$ such that
$$V = S \sqcup \Gamma_{\widetilde{G} }(S) \sqcup \Gamma_{\widetilde{G} }^2(S) \sqcup \dots \sqcup \Gamma_{\widetilde{G}}^k(S) = \bigsqcup_{i=0}^k \Gamma(S)_{\widetilde{G} }^i$$

The path taken through the Markov chain of possibilities will be the following. Select any $v' \in \Gamma_{\widetilde{G}}(S)$. For all $u \in V$, $f_{0}(u) \leq  M_0 $. Further $f_0(s) = M$ and $(v',s) \in E$ for some $s \in S$. Therefore
updating the graph $G$ on the node $v'$ produces 
        \begin{equation*}
            f_t(v) = f_1(v) = \begin{cases}
                M_0 & \text{if } v = v' \\
                f_{0}(v) & \text{if } v \neq v' \\
            \end{cases}
        \end{equation*}
The valuation $f_0$ is identical to $f_1$ except now we have one more vertex with value $M_0$, namely a vertex $v' \in \Gamma_{\widetilde{G}}(S)$. We can now repeat this process, one at a time and without repetition choose a vertex from $\Gamma_{\widetilde{G}}(S)$. Then one at a time and without repetition choose a vertex in $\Gamma^2_{\widetilde{G}}(S)$. Continue and repeat this until we have chosen all vertices in $\Gamma^k_{\widetilde{G}}(S)$. This ordering guarantees that every time a vertex is updated, its value is changed to $M_0$. This process provides a sequence of vertices to choose such that eventually the graph has a constant valuation $f_{con}$ with values $M_0$ on the vertices. Therefore there exists a path in the Markov chain of possibilities from $(G,f)$ to $(G,f_{con})$.

\end{proof}

\begin{theorem}
    The period of a strongly connected graph under the asynchronous maximum model is 1.
\end{theorem}
\begin{proof}
    The proof is identical to Theorem \ref{path_in_chain_implies_period_1}. But instead of invoking Lemma \ref{con_path}, invoke Lemma \ref{path_in_MC_SCC}.
\end{proof}

\section{Convergence}
In the analysis of undirected graphs we could leverage that $(u,v) \in E$ if and only if $(v,u) \in E$. This allowed us to define a strong edge. For a directed graph we have an analogous concept called a \textit{strong cycle}.

\begin{definition} \label{strong_cycle}
(Strong Cycle) Let $(G,f_t)$ be a directed graph with valuation $f_t$. Let $v_1,v_2,\dots,v_k \in V$ be a sequence of distinct vertices such that $f_t(v_i) = M_t$ for $1\leq i \leq k$. The sequence $v_1,v_2,\dots,v_k,v_1 $ is called a \textit{strong cycle} if $(v_i,v_{i+1}) \in E$ for all $1\leq i \leq k-1$ and $(v_{k},v_1) \in E$.
\end{definition}

A strong cycle is a cycle where every vertex in the cycle has the current maximum value. Note that in an undirected graph a strong edge is equivalent to a strong cycle with $k = 2$. Observe that once a strong cycle is formed, the values at the vertices in the cycle can never be changed. Again we can define an analogous notion to the strong edge set, which we call the \textit{strong cycle set}.

\begin{definition}\label{str_cycle_set}
    (Strong Cycle Set) Let $(G,f_t)$ be strongly connected graph with valuation $f_t$. Suppose $v \in V$ such that at least one of the following conditions hold;
    \begin{enumerate}
        \item $v$ is a member of a strong cycle
        \item There exists a walk $v_1,v_2,\dots v_k$ such that $v_1 = v$, $v_k$ is a member of a strong cycle and $f_t(v_i) = M_t$ for all $1\leq i \leq k$.
    \end{enumerate}
    That is, either $v$ is a member of a strong cycle or there is some walk from $v$ to a strong cycle where every vertex on the walk has the current maximum value. Define $C_t$ to be the largest set of vertices satisfying either conditions 1 or 2 (some vertices in $C_t$ may satisfy both). If the valuation is clear from context then write $C = C_t$
\end{definition}
\begin{example}
    In Figure \ref{sc_example}, Round 6 contains a strong cycle set where $|C| = 6$. There are three vertices in a strong cycle (the directed triangle) and three vertices on a path connected to a strong cycle.
\end{example}

\subsection{Potential Function}
Now we define a potential function on $(G,f_t)$. Let $g : (V \to [n]) \to \mathbb{Z}$ be given by

\begin{align*}
    & h(f_t) = |C_t| && \text{$C_t \subseteq V$ is the strong cycle set under valuation $f_t$} \\
\end{align*}

\begin{lemma}\label{potential_func_1}
    Let $(G,f_0)$ be a graph with initial valuation $f_0$. Update the valuations according to the asynchronous maximum model. The potential function $h$ is non-negative, non-decreasing (with respect to $t$) and bounded above by $n$.
\end{lemma}
\begin{proof}
    By definition $h(f_t)$ must be non-negative and bounded above by $n$. Therefore all that remains is to show $h(f_t)$ is non-decreasing with respect to $t$. To do this we show $C_t \subseteq C_{t+1}$ for all $t$, which implies $|C_t| \leq |C_{t+1}|$.
    \par Let $v'$ be the vertex chosen for the update from $f_t$ to $f_{t+1}$. We split into two cases depending on whether $v' \in C_t$ or not.
    \par \textbf{Case 1: } Suppose $v' \in C_t$. This implies $v'$ has the maximum value $M_t$ in round $t$ and there exists a $u \in V$ such that $f_t(u) = M_t$ and $v \sim u$. Therefore the value of $v'$ does not change during the update, so $f_{t+1} = f_t$ which implies $C_t = C_{t+1}$. 
    \par \textbf{Case 2:}
    Suppose $v' \notin C_t$. We now have to consider whether the vertex $v'$ has an out-neighbour in the set $C_t$ or not.
    \par \setlength{\leftskip}{0.5cm}  \textbf{Case 2.1:} If there exists a $w \in C_t$ such that $v' \sim w$ then when the update occurs, $v'$ will be added to the strong cycle set. Since only $v'$ changes value, all vertices in $C_t$ are also in $C_{t+1}$. This is because membership to these sets is determined by the values of the vertices (they need to be the current maximum) and the edges between them, neither of which change. Therefore $C_t \cup \{v'\} = C_{t+1}$.
    \par \setlength{\leftskip}{0.5cm}  \textbf{Case 2.2:} Assume there does not exist $w \in C_t$ such that $v' \sim w$. This could occur because of two reasons. Either $|C_t| = 0$ or $|C_t| > 0$ but $v'$ has no out-neighbours in $C_t$. We consider both cases below.
    \par \setlength{\leftskip}{1cm} \textbf{Case 2.2.1:} Suppose $|C_t| = 0$. Then $C_t = \emptyset$. We trivially have that $C_t = \emptyset \subseteq C_{t+1}$.
    \par \setlength{\leftskip}{1cm} \textbf{Case 2.2.1:} Suppose $|C_t| > 0$ and that for all $w \in \Gamma(v')$, $w \notin C_t$. That is $v'$ has no out neighbours in the strong cycle set. Note that for all $v \in C_t$, we also have $v \in C_{t+1}$. This is because membership to $C_t$ only depends on the values on the vertices $v \in C_t$ and the edges between them. A subtle point is that $v'$ cannot update to a value larger than $M_t$. So the vertices in $C_t$ indeed still have the current maximum value and $M_t = M_{t+1}$. Therefore we have that $C_t \subseteq C_{t+1}$.

\par \setlength{\leftskip}{0.3cm} Although this completes the proof we can also say that in Case 2.2.1, $v'$ will join the strong cycle set if by updating it results in the creation of a strong cycle. This is because $v'$ has no out-neighbours in $C_t$, so updating it cannot result in Condition 2 of Definition \ref{str_cycle_set} being satisfied. However potentially $v'$ lies on a directed cycle $W$, where $f_t(v') < M_t$ an $f_t(w) = M_t$ for $w \in W$. Then updating $v'$ will result in the formation of a strong cycle.
\end{proof}

The potential function is useful as now we can analyse the convergence time in terms of $h$. We analyse the following two stages of convergence.

\begin{enumerate}
    \item How long until a strong cycle is formed?
    \item Given there exists a strong cycle, how long until $|C| = n$?
\end{enumerate}
We will call these Phase 1 and Phase 2 respectively.

\subsection{Phase 1: Strong Cycle}

\begin{example}
Please refer to Figure \ref{sc_example} in the Appendix. The red nodes are selected for update. Initially $f_0$ has no strong cycle and the potential function $h(f_0) = 0$. Phase 1 is completed in Round $6$.
\end{example}

\begin{definition} (Maximal Chain of the Cycle $G[U]$)
    Let $U \subseteq V$ be such that $G[U]$ is a cycle. The maximal chain is the longest path $P_U = (u_1,u_2,\dots, u_k)$ in $G[U]$ such that $f_t(u_i) = M_t$ for all $1\leq i\leq k$. We write $P_U$ to denote the maximal chain and the subscript $U$ reminds us that the maximal chain depends on the cycle $U$.
\end{definition}

\begin{example}
    In Figure \ref{sc_example}, after Round 2 there is a maximal chain of length $3$ which is composed of vertices with value $5$ in the cycle of length $6$.
\end{example}

It will be useful for us to select a certain type of maximal chain in the graph $G$. Suppose we select all cycles that contain $M_t$ and out of these cycles restrict to the ones with the smallest length. Then out of these small cycles pick one with the largest maximal chain. Note that there may be multiple small cycles each with the same length of maximal chain so we must make a choice in this process. This maximal chain is captured in the following definition.

\begin{definition}\label{orbit_def}
    (Orbit)
    Let $G$ be a strongly connected graph. For each $v \in V$, define $b(v)$ as the length of the smallest cycle containing $v$. Since $G$ is strongly connected every vertex is contained in a cycle so $b(v)$ is well defined. The \textit{Orbit} of $G$ is
    $$b(G) := \max_{v \in V} b(v)$$
\end{definition}
A clear upper bound on the orbit is the circumference $c$ of the graph, that is the size of the largest cycle in $G$. However we potentially could have $b(G) \ll c$, consider the complete graph $K_n$ where $b(G) = 2 \ll n = c$.

\begin{definition}
    (Max-min Chain in $G$)
    Let $G$ be a strongly connected directed graph with valuation $f_t$ and maximum value $M_t$. Let $W \subseteq V$ be the set of vertices in $G$ with value $M_t$.
    Let
    $$\mathcal{U} = \{U \subseteq V | G[U] \text{ is a cycle containing a } w\in W\}$$
    be the set of cycles containing values in $M_t$.
    Now define $\alpha_{min} = \min_{w \in W}{b(w)}$ to be the size of the smallest cycle containing $M_t$. Let 
    $$\mathcal{U}_{min} = \{U \in \mathcal{U} | |U| = \alpha_{min}\}.$$
    be the set of the smallest cycles containing the value $M_t$.
    Now let $\beta_{max} = \max_{U \in \mathcal{U}_{min}}{|P_U|}$ be the size of the largest maximal chain of all the cycles in $\mathcal{U}_{min}$.
    
    A \textit{max-min chain} in $G$ is a path denoted by $P$ where
    $$P \in \{P_U  | |P_U| = \beta_{max} \text{ and } U \in \mathcal{U}_{min} \}.$$
\end{definition}
Note that there always exists a max-min chain in $G$. There may be multiple cycles satisfying the requirements listed in the definition but we can just pick one of them. The name max-min should suggest that we are looking at a maximal chain in some minimal cycle. In the bound for Phase 1 we will look at a min-max chain $P$ contained in some cycle $U$ and analyse the expected number of rounds until $|P| = |U|$.

\par

With reference to Figure \ref{sc_example}, consider the cycle of length $3$ in Round 1, this is a minimal cycle containing $M_1 = 5$. The min-max chain in $G$ is the vertex valued $5$ contained in the cycle of length $3$. Let $U$ be the vertices containing this min-max chain $P = P_U$. We can then consider $G[U]$ and note that the maximal chain of $G[U]$ can either increase or decrease in length in each round. It may increase until its length $|P_U| = |U|$ and in this case a strong cycle is formed. Otherwise it may decrease in length. For example in Round $3$, $|P_U| = 2$ while in Round $4$, $|P_U| = 1$. In Round $6$ we have $|P_U| = |U| = 3$ and the formation of a strong cycle. 

\par However we may have that $|P_U| = 0$ after some update. In this case we can look towards another cycle $U'$ containing a max-min chain in $G$. We then try to analyse the expected number of rounds until $|P_{U'}| = |U'|$.

\par In the above argument we eventually want to have $|P_U| = |U|$ for some cycle $U$ containing a max-min chain in $G$. Once we \textbf{fix} a cycle $U$, the problem is similar to Gambler's Ruin with two players $A$ and $B$. Player $A$ has $|P_U|$ dollars and Player $B$ has $|U| - |P_U|$ dollars. In each round there is some probability that Player $A$ wins one dollar from Player $B$, there is some Probability that Player $A$ loses one dollar to Player $B$ and there is some probability that neither of them win or lose. We make the following argument in the next theorem to characterise the expected number of rounds until a strong cycle is formed.

\begin{theorem}\label{strong_cycle_exp}
    Let $(G,f)$ be a strongly connected graph with valuation $f$ and orbit $b = b(G)$. Let $X$ be a random variable denoting the number of rounds until a strong cycle is formed. Then $\E[X] = \bigO(nb^2)$.
\end{theorem}
\begin{proof}
   Firstly, we argue that $\E[X]$ can be upper bounded by the following game of Gambler's Ruin.  Given the valuation $f$, let $P_U \subseteq U \subseteq V$ denote a max-min chain $P_U$ contained in the cycle $U$. We will now fix attention to the cycle $U$ and the max-min chain $P_U$. Note that as we perform updates, the path $P_U$ may cease to be a min-max chain. This does not matter since we only require that when we \textit{initially} pick $P_U$ and a cycle $U$ that it is a min-max chain.
    \par Initially, Player $A$ has $|P_U| \in \{1,2,\dots, |U|\}$ dollars and Player $B$ has $|U| - |P_U|$ dollars. Player $A$ models the length of the path $P_U$. A strong cycle is formed exactly when $|P_U| = |U|$ for some cycle $U$ containing $P_U$. This occurs exactly when $B$ is ruined. We can now perform updates from $f = f_0$ according to the asynchronous maximum model.
    \par We argue that Player $A$ can never be ruined. We have already fixed a cycle $U$ and are considering a path $P_U$. Suppose in Round $t$ we have $|P_U| = 1$. Now in round $t+1$ suppose we select a vertex in order to make $|P_U| = 0$. (This is possible, consider Figure \ref{sc_example} in Round $t = 5$, if the model selected the vertex valued $5$ in the cycle of length $3$ then this scenario would have occurred in Round $t+1$). However, now in round $t+1$ we immediately switch to another min-max chain $P_{U'}$ (which depends on $f_{t+1}$) contained in a new cycle $U'$. Now we fix attention to this cycle $U'$. We have that $|P_{U'}| = \{1,2,\dots,|U'|\}$. Since the min-max chain depends on the valuation then the size of $U'$ may be larger or smaller or the same compared to the size of $U$. The important observation is that in Round $t+1$, $|P_{U'}| \geq 1$ and we can consider Player $A$ as owning $|P_{U'}|$ dollars and Player $B$ as owning $|U'| - |P_{U'}|$ dollars. For any max-min chain $P$ in a cycle $U$, we have that $|U| \leq b(G)$. In the worst case, Player $A$ would start with $1$ dollar and need to obtain $b(G)$ dollars. For each round given that $P_U$ is the max-min chain in that round we have,
    $$\mathcal{P}(|P_U| \text{ increases}) = \frac{1}{n}$$
    $$\mathcal{P}(|P_U| \text{ decreases}) = 0 \text{ or } \frac{1}{n}$$
    $$\mathcal{P}(|P_U| \text{ does not increase or decrease}) = \frac{n-1}{n} \text{ or } \frac{n-2}{n}$$
    Then the worst case probabilities from the perspective of Player $A$ are
    $$\mathcal{P}(|P_U| \text{ increases}) = \frac{1}{n}$$
    $$\mathcal{P}(|P_U| \text{ decreases}) = \frac{1}{n}$$
    $$\mathcal{P}(|P_U| \text{ does not increase or decrease}) = \frac{n-2}{n}$$
    For example, in Figure \ref{sc_example}, in Round 3 the probability that the maximal chain in the cycle of length $6$ decreases is zero. However if situations like these do not occur then it is worse for Player $A$.
    \par Therefore, the expected number of rounds until we obtain a strong cycle can be bounded by a game of Gamblers Ruin where Player $A$ starts with $k \in \{1,2,\dots, b(G)\}$ dollars and can never be ruined, and Player $B$ starts with $l \in \{0,1,2,\dots, b(G)-1\}$ dollars such that $k + l = b(G)$. We assume $k = 1$ in the worst case. There is a $\frac{1}{n}$ chance of Player $A$ giving a dollar to Player $B$, a $\frac{1}{n}$ chance of Player $B$ giving a dollar to Player $A$ and a $\frac{n-2}{n}$ chance of neither player giving money to the other. Lemma \ref{gamblers_ruin_expectation} completes the proof that $\E[X] = \bigO(nb^2)$.
\end{proof}

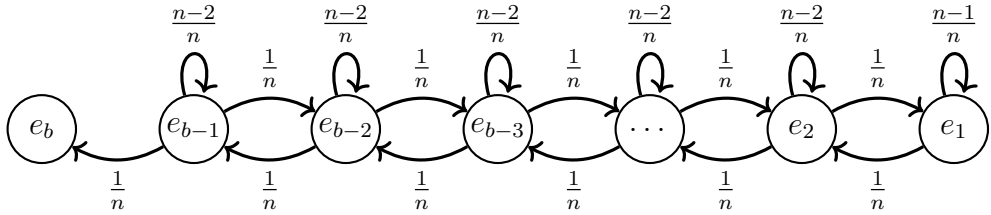
\begin{figure}[h]
\centering
\begin{tikzpicture}[
every edge/.style = {->, draw=black,very thick},
 vrtx/.style args = {#1/#2}{%
      circle, draw, thick, fill=white,
      minimum size=9mm, label=#1:#2}
                    ]
\node(A) [vrtx= center/$e_b$] at (0, 0) {};
\node(B) [vrtx= center/$e_{b-1}$] at (2, 0) {};
\node(C) [vrtx= center/$e_{b-2}$] at (4, 0) {};
\node(D) [vrtx= center/$e_{b-3}$] at (6, 0) {};
\node(E) [vrtx= center/$e_{2}$] at (10, 0) {};
\node(F) [vrtx= center/$e_{1}$] at (12, 0) {};
\node(M) [vrtx = center/$\dots$] at (8, 0) {};

\path   (B) [bend left] edge node[yshift = -4mm] {$\frac{1}{n}$} (A);
\path   (B) [bend left] edge node[yshift = 4mm] {$\frac{1}{n}$} (C);
\path   (C) [bend left] edge node[yshift = -4mm] {$\frac{1}{n}$} (B);
\path   (D) [bend left]  edge node[yshift = -4mm] {$\frac{1}{n}$} (C);
\path   (D) [bend left]  edge node[yshift = 4mm] {$\frac{1}{n}$} (M);
\path   (M) [bend left]  edge node[yshift = -4mm] {$\frac{1}{n}$} (D);
\path   (M) [bend left]  edge node[yshift = 4mm] {$\frac{1}{n}$} (E);
\path   (E) [bend left]  edge node[yshift = -4mm] {$\frac{1}{n}$} (M);
\path   (C) [bend left]  edge node[yshift = 4mm] {$\frac{1}{n}$} (D);
\path   (E) [bend left]  edge node[yshift = 4mm] {$\frac{1}{n}$} (F);
\path   (F) [bend left]  edge node[yshift = -4mm] {$\frac{1}{n}$} (E);
\path   (F) edge [loop above] node {$\frac{n-1}{n}$} (F);
\path   (E) edge [loop above] node {$\frac{n-2}{n}$} (E);
\path   (B) edge [loop above] node {$\frac{n-2}{n}$} (B);
\path   (C) edge [loop above] node {$\frac{n-2}{n}$} (C);
\path   (D) edge [loop above] node {$\frac{n-2}{n}$} (D);
\path   (M) edge [loop above] node {$\frac{n-2}{n}$} (M);

\end{tikzpicture} 
\caption{Transition probabilities in the Gamblers Ruin Game for Player $A$ where $A$ has $b-k$ dollars.}
\label{exp_game}
\end{figure}

\begin{lemma} \label{gamblers_ruin_expectation} 
    The expected number of rounds in the game of gamblers ruin described at the end of Theorem \ref{strong_cycle_exp} is in $\bigO(nb^2)$.
\end{lemma}
\begin{proof}
    Let $e_k$ be the expected number of rounds until Player $A$ has $b$ dollars given that Player $A$ starts with $k$ dollars. We have that $k \in \{1,2,\dots, b(G)\}$ and want to find $e_1$. Clearly $e_b = 0$. Please refer to Figure \ref{exp_game}. For $1\leq k \leq b-2$, we can derive the following equation by conditioning on the events that Player $A$ gives money to $B$, gains money from $B$ or neither wins or loses.
    \begin{equation}\label{exp_recc}
        e_{b-k} = 1 + \frac{1}{n}e_{b-(k-1)} + \frac{1}{n}e_{b-(k+1)} + \frac{n-2}{n}e_{b-k}
    \end{equation}
    Now we claim that for $0 \leq k \leq b-2$,
    \begin{equation}\label{recc_ans}
        e_{b-k} = \frac{kn}{2} + \frac{k}{k+1}e_{b-(k+1)} 
    \end{equation}
    where this equation comes from \cite{hthesis}. The proof is by induction on $k$. For $k = 0$, $e_b = 0$. Now suppose the claim holds up to $k-1$. Since we have $1\leq k \leq b-2$ then Equation \ref{exp_recc} gives
    \begin{align*}
        & e_{b-k} = 1 + \frac{1}{n}e_{b-(k-1)} + \frac{1}{n}e_{b-(k+1)} + \frac{n-2}{n}e_{b-k} && \text{} \\
        & e_{b-k} = 1 + \frac{1}{n}\left(\frac{(k-1)n}{2} + \frac{k-1}{k}e_{b-k}  \right) + \frac{1}{n}e_{b-(k+1)} + \frac{n-2}{n}e_{b-k} && \text{Induction Hypothesis} \\
        & \frac{2}{n}\cdot e_{b-k} = 1 + \frac{1}{n}\left(\frac{(k-1)n}{2} + \frac{k-1}{k}e_{b-k}  \right) + \frac{1}{n}e_{b-(k+1)}  && \text{} \\
        & \frac{2}{n}\cdot e_{b-k} = 1 + \frac{(k-1)}{2} + \frac{k-1}{kn}e_{b-k}  + \frac{1}{n}e_{b-(k+1)}  && \text{} \\
        & \left(\frac{2}{n}-\frac{k-1}{kn}\right)  e_{b-k} = 1 + \frac{(k-1)}{2} + \frac{1}{n}e_{b-(k+1)}  && \text{} \\
        & \left(\frac{k+1}{kn}\right)  e_{b-k} = 1 + \frac{(k-1)}{2} + \frac{1}{n}e_{b-(k+1)}  && \text{} \\
        & \left(\frac{k+1}{kn}\right)  e_{b-k} =  \frac{k+1}{2} + \frac{1}{n}e_{b-(k+1)}  && \text{} \\
        &  e_{b-k} =  \frac{kn}{2} + \frac{k}{k+1}e_{b-(k+1)}  && \text{} \\
    \end{align*}
Therefore, by induction we have that Equation \ref{recc_ans} holds for all $0\leq k \leq b-2$. We can derive a recurrence for when $k = b-1$ by again conditioning on whether $A$ wins, loses or the round is a draw. We obtain
$$e_{1} = 1 + \frac{1}{n}e_2 + \frac{n-1}{n}e_1.$$
Therefore 
$$e_1 = n + e_2.$$
Using the claim we have that $e_2 = e_{b-(b-2)} = \frac{(b-2)n}{2} + \frac{b-2}{b-1}e_1$. Combining these equations gives
$$e_1 = n + e_2 =  n +  \frac{(b-2)n}{2} + \frac{b-2}{b-1}e_1.$$
Solving for $e_1$ gives
$$e_1 = \frac{nb(b-1)}{2} = \bigO(nb^2)$$
which is the expected number of rounds until $A$ has $b$ dollars given $A$ starts with $1$ dollar.

\end{proof}

We remark that for undirected graphs $b = 2$, and the time to a strong cycle (which for $b = 2$ is a strong edge) is $\bigO(n)$.

\subsection{Phase 2: Strong Cycle Set}
Now we try to bound the worst case convergence time in Phase 2. That is we assume a strong cycle exists in the graph and bound the number of rounds until the strong cycle set is the entire graph.

\begin{theorem}\label{strong_cycle_set_bound}
    Let $G$ be a directed graph with non-empty strong cycle set $C$. For $2\leq k\leq n$, let $Z$ be a random variable denoting the number of rounds until $|C| = n$ given $|C| = k$. Then
    $$\E[Z] = \bigO(\frac{n}{\phi'}\log n)$$
    Where $\phi'= \min(\phi_{out},\phi_{in})$.
\end{theorem}

\begin{proof}
    We will split process into Phase 2.1 and Phase 2.2. In Phase 2.1 we bound the expected number of rounds until $|C| = \lfloor \frac{n}{2} \rfloor + 1 $. In Phase 2.2 we bound the expected number of rounds from $|C| = \lfloor \frac{n}{2} \rfloor + 1$ until $|C| = n$. We call an update successful if it increases the quantity $|C|$.

    \textbf{Phase 2.1:}
    Let $Z$ be a random variable denoting the number of rounds until $|C| = \lfloor \frac{n}{2} \rfloor + 1 $ given $|C| = k$. We assume that $k \leq \lfloor \frac{n}{2} \rfloor + 1$ as otherwise $\E[Y] = 0 $. Let $Z_i$ be a random variable denoting the number of rounds until $|C| = i$ given that $|C| = i-1$. Each random variable $Y_i$ stochastically dominates a geometric random variable $Y_i'$ with success probability $\frac{\phi_{in} \cdot (i-1)}{n}$. This is because there are at least $\phi_{in}\cdot |C|$ vertices that have edges directed towards the set $C$. 
    
  Therefore
$$\E[Z] = \sum_{i=k+1}^{\lfloor \frac{n}{2} \rfloor + 1} \E[Z_i] \leq \sum_{i=3}^{\lfloor \frac{n}{2} \rfloor + 1} \E[Z_i]$$
Using $\E[Z_i] \leq \E[Z_i']$ then gives
  
    $$ \leq \sum_{i=3}^{\lfloor \frac{n}{2} \rfloor + 1} \E[Z_i'] =\sum_{i=3}^{\lfloor \frac{n}{2} \rfloor + 1} \frac{n}{\phi_{in} (i-1)} \leq \frac{n}{\phi_{in}} \sum_{j=1}^{n} \frac{1}{j} = \bigO(\frac{n}{\phi_{in}}\log n).$$

    \textbf{Phase 2.2:}

     Here we assume $\lfloor \frac{n}{2} \rfloor + 1 \leq |C| < n $. We partition $G$ into the sets $C, W = \Gamma_{\widetilde{G}}(C)$ and $U = V \setminus (S \cup W)$. The diagram is similar to Figure $\ref{fig:phase3}$. In the worst case an update will only be successful when a vertex $v \in W$ is selected. We will first bound the number of rounds to shrink $|U|$ to zero then the number of rounds to shrink $|W|$ to zero.

\par Let $Y$ be a random variable denoting the number of rounds until $|U| = 0$. Let $Y_i$ be the number of rounds until $|U| = i$ given $|U| = i + 1$. We make the assumption that in the worst case $|U|$ can decrease by at most one during an update. Let $|U| = m \leq \lfloor \frac{n}{2} \rfloor -2 \leq n-1$ be the initial size of $U$. Then $Y = \sum_{i = 0}^{m} Y_i$. In order to decrease $|U|$ by one, we choose a vertex in $\Gamma_G(U)$. This occurs with probability $\frac{\Gamma(U)}{n}$. Further, $\phi_{out}|U| \leq \Gamma(U)$. This implies $\frac{\Gamma(U)}{n} \geq \frac{\phi_{out} |U|}{n}$. Therefore each $Y_i$ stochastically dominates a geometric random variable $Y_i'$ with success probability $\frac{\phi_{out} |U|}{n} = \frac{\phi_{out}(i+1)}{n}$.
We have that
$$\E[Y] = \sum_{i = 0}^{m-1} \E[Y_i]  \leq \sum_{i = 0}^{m-1} \E[Y_i'] = \sum_{i = 0}^{m-1} \frac{n}{\phi_{out}(i+1)}   $$

$$ \leq \sum_{i = 0}^{\lfloor \frac{n}{2} \rfloor -3} \frac{n}{\phi_{out}(i+1)} = \frac{n}{\phi_{out}} \sum_{i = 0}^{n-1} \frac{1}{i+1} \leq \frac{n}{\phi_{out}}H_n = \bigO(\frac{n}{\phi_{out}}\log n).$$

\par Now we assume $|U| = 0$ and consider the number of rounds to shrink $|W| = |\Gamma_{\widetilde{G}}(C)|$ to $0$. Let $|W| = p \leq \lfloor \frac{n}{2} \rfloor - 1 \leq n$ be the initial size of $W$. Since $|U| = 0$ then the vertices of $G$ are either in $C$ or $W = \Gamma_{\widetilde{G}}(C)$. Let $R$ be a random variable denoting the number of rounds until $|W| = 0$ given $|W| = p$. Let $R_i$ denote the number of rounds until $|W| = i$ given that $|W| = i+1$. Then $R = \sum_{i = 0}^{p-1} R_i$. Each $R_i$ is a geometric random variable with success probability $\frac{|W|}{n} = \frac{i+1}{n}$. Then the expected number of rounds until a successful update is $\frac{n}{i+1}$. Therefore
$$\E[R] = \sum_{i = 0}^{p-1} \E[R_i] = \sum_{i = 0}^{p-1} \frac{n}{i+1}  $$
Using $p \leq \lfloor \frac{n}{2} \rfloor - 1  \leq n$ gives
$$ \leq \sum_{i = 0}^{\lfloor \frac{n}{2} \rfloor - 2} \frac{n}{i+1} \leq n \sum_{j=1}^n \frac{1}{i} = n H_n = \bigO(n\log n) = \bigO(\frac{n}{\phi_{out}}\log n).$$

The last equality arises since $\phi_{out} \leq 5$. Combining all of the phases gives us that the number of rounds until $|C| = n$ given $|C| = k$ for $2 \leq k \leq n$ is

$$\bigO(\frac{n}{\phi_{in}}\log n) + \bigO(\frac{n}{\phi_{out}}\log n)+\bigO(\frac{n}{\phi_{out}}\log n) = \bigO(\frac{n}{\phi'}\log n).$$
    
\end{proof}

\begin{theorem}\label{conv_time_sc_graph}
    Let $G$ be a strongly connected graph with Orbit $b$ and let $\phi' = \min(\phi_{out},\phi_{in})$. The convergence time of $G$ under the asynchronous maximum model is
    $$\mathcal{T}(G) = \bigO(nb^2 + \frac{n}{\phi'} \log n).$$
\end{theorem}
\begin{proof}
    Let $f$ be any valuation of $G$. We first need to form a strong cycle set $C$ and then we need to increase $|C|$ until $|C| = n$. If $|C| = 0$ then by Theorem \ref{strong_cycle_exp} in expectation it takes $\bigO(nb^2)$ rounds until $|C| > 0$. By Theorem \ref{strong_cycle_set_bound} in expectation we need $\bigO(\frac{n}{\phi'} \log n)$ rounds until $|C| = n$ given $|C| > 0$. The convergence time $\mathcal{T}(G)$ can therefore be bounded by the sum of these two processes.
\end{proof}

We remark that for undirected graphs, $b = 2$ and $\phi = \phi_{out} = \phi_{in}$ so the bound in Theorem \ref{conv_time_sc_graph} reduces to $\bigO(\frac{n}{\phi}\log n)$ which agrees with our earlier analysis.

\begin{question}
    Is it possible to formulate $\T(G)$ in terms of $\phi_{in}$ (or respectively $\phi_{out}$) only?
\end{question}
If we wanted to show the answer to the above is negative, it would be enough to provide a graph with large (near constant) $\phi_{out}$ but arbitrarily small $\phi_{in}$ (or vice-versa). This would indicate that the parameter $\phi'$ is necessary to obtain a good characterisation of the convergence times in terms of the vertex expansions.

\begin{question}
    Do there exist strongly connected graphs $G$ with $\phi_{in}$ arbitrarily small ($\phi_{in} = o(1)$) while $\phi_{out}$ is approximately a constant ($\phi_{out} = \theta(1)$).
\end{question}

\begin{question}
    Can the parameter $b$ be written in terms of $\phi_{in}$ (respectively $\phi_{out}$). For example we would expect that every vertex in a graph with good vertex out-expansion to be contained in a small cycle.
\end{question}

%% file: chapters/conclusion.tex
We introduced the notion of an iterative graph model. This captures \textit{any} model which updates vertices of a graph $G$ according to some rule. If there are only a finite number of possible states a graph $G$ can update to we can define the Markov Chain of Possibilities $\mathcal{G}$. This allows for a rigorous definition of the period and convergence time in terms of absorbing components of $\mathcal{G}$ and the number of rounds taken to reach them in the worst case. The asynchronous maximum model is introduced as one such iterative graph model.

\par Chapter 3 is concerned with undirected graphs $G$. It is shown that every absorbing component in $\mathcal{G}$ has size $1$ and contains a constant valuation state. The convergence time is bounded by looking at two phases. Phase 1 is concerned with the time until a strong edge is formed and Phase 2 examines the expected number of rounds until the values at this strong edge propagate to the rest of the graph. It is shown the convergence time is lower bounded by $\Theta(n \log n)$ and bounded above by $\bigO(n^2)$. Further these bounds are tight by example of $K_n$ and $P_n$. We can potentially improve the upper bound for a large class of graphs by introducing the vertex expansion. We show $\mathcal{T}(G) = \bigO(\frac{n}{\phi}\log n)$ which is better than $\bigO(n^2)$ when $\phi = o(\frac{\log n}{n})$. Further, we show that with high probability the process converges in $\bigO((\frac{n}{\phi}\log n)(g(n)))$ rounds where $g(n)$ is any function such that $\frac{1}{g^2(n)} \to 0$ as $n\to \infty$. This result shows that the number of rounds until convergence only exceeds the expected number of rounds until convergence by a factor of $g(n)$.

\par In Chapter 4 we study the model for strongly connected directed graphs. The period of the model is again shown to be $1$. We again bound the convergence time by splitting into two phases. Phase 1 is concerned with the time taken to reach a strong cycle. The expected number of rounds until a strong cycle is formed is shown to be bounded above by a game of gamblers ruin. Phase 2 is concerned with the expected number of rounds until the strong cycle set has propagated through the entire graph. The expected number of rounds of both phases is shown to be in $\bigO(nb^2 + \frac{n}{\phi'}\log n)$.

\section{Future Research}

The study of the asynchronous maximum model leaves us with some interesting questions for further research. The first clear extension would be to generalise the results for weakly connected graphs. In this case the period is still $1$ however the absorbing states in the Markov Chain are not constant valuation states. An example of such a graph in an absorbing state is given in Figure \ref{fig:wk_ex}.

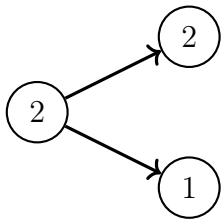
\begin{figure}[h!]
    \centering
    \begin{tikzpicture}[
every edge/.style = {draw=black,very thick},
 vrtx/.style args = {#1/#2}{%
      circle, draw, thick, fill=white,
      minimum size=8mm, label=#1:#2}]
\node(A) [vrtx= center/2] at (0, 0) {};
\node(B) [vrtx= center/2] at (2, 1) {};
\node(C) [vrtx= center/1] at (2, -1) {};

\path   (B) edge[<-] (A)
        (C) edge[<-] (A);
\end{tikzpicture}
    \caption{The values on this weakly connected graph will not change}
    \label{fig:wk_ex}
\end{figure}

\par The convergence time for a weakly connected graph $G$ could be analysed as follows. Firstly, partition $G$ into maximal strongly connected components. This partition induces a partial ordering of the maximal strongly connected components. The convergence time can then be bounded by the time taken for each component to converge.

\par In the strongly connected analysis we are unsure whether the parameter $\phi' = \min\{\phi_{in},\phi_{out}\}$ is needed. Is it possible to formulate $\mathcal{T}(G)$ in terms of $\phi_{in}$ only? To answer this question in the negative, and an interesting problem in its own right, would be to construct a graph $G$ (which is strongly connected) where one of these parameters is large while the other is arbitrarily small. Does such a $G$ exists?\newline

\par Further, consider the strongly connected case. We know the absorbing states are constant valuation states. The next natural question to ask is what values are on the vertices when the process converges. More precisely, let $f_0$ be an initial valuation which takes values in $S \subseteq [n]$. For each $s \in S$, what is the probability that the final valuation is of the form $f_{con}(x) = s$ for all $x \in V$? We will call the valuation it equals at convergence the absorbing valuation of $(G,f_0)$.

\begin{figure}[h!]
\centering
    \begin{tabularx}{0.95\textwidth}{*{3}{>{\centering\arraybackslash}X}}
\begin{tikzpicture}[
every edge/.style = {draw=black,very thick},
 vrtx/.style args = {#1/#2}{%
      circle, draw, thick, fill=white,
      minimum size=8mm, label=#1:#2}
                    ]
\node(A) [vrtx= center/1] at (0, 0) {};
\node(B) [vrtx= center/1] at (-1, 1.6) {};
\node(C) [vrtx= center/1] at (0,3.2) {};
\node(D) [vrtx= center/1] at (1.8,3.2) {};
\node(E) [vrtx= center/1] at (2.8, 1.6) {};
\node(F) [vrtx= center/1] at (1.8, 0) {};
\path   (A) edge[->] (B)
        (B) edge[->] (C)
        (C) edge[->] (D)
        (D) edge[->] (E)
        (E) edge[->] (F)
        (F) edge[->] (A);
\end{tikzpicture}
    \caption*{$G_1$}  
    &   
\begin{tikzpicture}[
every edge/.style = {draw=black,very thick},
 vrtx/.style args = {#1/#2}{%
      circle, draw, thick, fill=white,
      minimum size=8mm, label=#1:#2}
                    ]
\node(A) [vrtx= center/1] at (0, 0) {};
\node(B) [vrtx= center/2] at (-1, 1.6) {};
\node(C) [vrtx= center/3] at (0,3.2) {};
\node(D) [vrtx= center/4] at (1.8,3.2) {};
\node(E) [vrtx= center/5] at (2.8, 1.6) {};
\node(F) [vrtx= center/6] at (1.8, 0) {};
\path   (A) edge[->] (B)
        (B) edge[->] (C)
        (C) edge[->] (D)
        (D) edge[->] (E)
        (E) edge[->] (F)
        (F) edge[->] (A);
\end{tikzpicture}
\caption*{$G_2$}  
&
\begin{tikzpicture}[
every edge/.style = {draw=black,very thick},
 vrtx/.style args = {#1/#2}{%
      circle, draw, thick, fill=white,
      minimum size=8mm, label=#1:#2}
                    ]
\node(A) [vrtx= center/1] at (0, 0) {};
\node(B) [vrtx= center/2] at (-1, 1.6) {};
\node(C) [vrtx= center/2] at (0,3.2) {};
\node(D) [vrtx= center/2] at (1.8,3.2) {};
\node(E) [vrtx= center/1] at (2.8, 1.6) {};
\node(F) [vrtx= center/2] at (1.8, 0) {};
\path   (A) edge[->] (B)
        (B) edge[->] (C)
        (C) edge[->] (D)
        (D) edge[->] (E)
        (E) edge[->] (F)
        (F) edge[->] (A);
\end{tikzpicture}
    \caption*{$G_3$}  
 \end{tabularx}
\caption{Three graphs with different initial valuations}
\label{fig:example_prob}
\end{figure} 

\par In some cases the answer is trivial. For example, $G_1$ in Figure \ref{fig:example_prob} with probability $1$ the absorbing valuation is the constant valuation sending each vertex to the value $1$. For $G_2$ potentially we should expect that all constant valuations (sending all vertices to a value in $[6]$) are equally likely to be the absorbing valuation. For $G_3$ the analysis would become more difficult but we expect the probability to depend on the following two criteria:
\begin{enumerate}
    \item The initial density of the value in the graph. For example the density of the value $2$ in $G_3$ is $\frac{4}{6}$. We would expect it is more likely for the absorbing valuation to be the valuation sending every vertex to $2$.
    \item The structure of the graph in relation to the value.
\end{enumerate}

The study of the asynchronous maximum model leads to interesting proof techniques which could potentially generalise to other iterative graph models. The questions we are left with also provide some potential avenues for future research.

%% file: chapters/appendix.tex
\section{Valuations}\label{val_app}
In this section we show that our definition of valuations as functions into $[n]$ does not lost any generality in the asynchronous maximum model. This section should be treated as stand alone and notation introduced here (such as a valuation and valuation family) should not be confused with the definitions given in the body of the paper.

\begin{definition}\label{config}
    (Configurations) Let $C$ be a non-empty, totally ordered set. We call $C$ a configuration.
\end{definition}

\begin{definition} (Valuation in terms of $C$)
Let $G = (V,E)$ be a graph (directed or undirected) with vertex set $V$ and edge set $E$. Fix $t \in \Z$ such that $t \geq 0$. A \textit{valuation} of the graph $G$ is a function $f_t : V \to C$.
\end{definition}

We note the above definition is more general that the one used in the paper, since now $C = [n], \N, \R$ is allowed. One may naturally ask if the set $C$ effects the asynchronous maximum model. For example is there a difference between choosing $C = \R$ and $C = \N$?. Given some mild assumptions on the cardinality of $C$ the answer is negative as we now show. This will justify our choice of $C = [n]$ in the paper.

\begin{definition}\label{iso_of_val_def}
    (Isomorphism of Valuations) Fix a graph $G = (V,E)$. Let $C,C'$ be configurations. Let $f : V \to C$ and $g : V \to C'$ be two valuations. An \textit{Isomorphism of Valuations} is a map $\alpha : V \to V$ such that the following hold:
    \begin{enumerate}
        \item $\alpha$ is a graph automorphism. That is $\alpha : V \to V$ such that $(u,v) \in E$ if and only if $(\alpha(u),\alpha(v)) \in E$.
        \item For all vertices $u,v \in V$, $f(u) < f(v)$ (resp. $=, >$) if and only if $g(\alpha(u)) < g(\alpha(v))$ (resp. $=, >$).
    \end{enumerate}
\end{definition}

\begin{example}\label{iso_of_val_ex}
    (Isomorphism of Valuations) Please consider the graph $G$ in Figure \ref{graph_1} with vertex labels $v_1,\dots, v_4$. 
\end{example}

There is an automorphism $\alpha: V \to V$ given by
    $$\alpha(v_1) = v_2$$
    $$\alpha(v_2) = v_1$$
    $$\alpha(v_3) = v_3$$
    $$\alpha(v_4) = v_4.$$
Now consider the following two valuations $f$ and $g$ given in Figure \ref{graphs_ex_2}. We have that $\alpha$ is an isomorphism of valuations. Note that the automorphism $\alpha$ swaps $v_1$ and $v_2$ so it does indeed preserve orderings between all vertices. Even though $f$ contains numbers from $\R$ and $g$ contains numbers in $\N$, we expect the asynchronous maximum model to update them in the same way. Therefore the isomorphism of valuations preserves the important information in the valuation, namely the relative ordering of the vertices.

\begin{figure}[h!]
    \centering
    \begin{tikzpicture}[
every edge/.style = {draw=black,very thick},
 vrtx/.style args = {#1/#2}{%
      circle, draw, thick, fill=white,
      minimum size=8mm, label=#1:#2}
                    ]
\node(A) [vrtx= center/$v_1$] at (0, 1) {};
\node(B) [vrtx= center/$v_2$] at (0, -1) {};
\node(C) [vrtx= center/$v_3$] at (1.5, 0) {};
\node(D) [vrtx= center/$v_4$] at (3, 0) {};

\path   (A) edge (B)
        (A) edge (C)
        (C) edge (B)
        (C) edge (D);
\end{tikzpicture}
    \caption{Graph $G$}
    \label{graph_1}
\end{figure}

\begin{figure}[h!]
\centering
    \begin{tabularx}{0.8\textwidth}{*{2}{>{\centering\arraybackslash}X}}
\begin{tikzpicture}[
every edge/.style = {draw=black,very thick},
 vrtx/.style args = {#1/#2}{%
      circle, draw, thick, fill=white,
      minimum size=8mm, label=#1:#2}
                    ]
\node(A) [vrtx= center/$\frac{1}{2}$] at (0, 1) {};
\node(B) [vrtx= center/$\frac{1}{3}$] at (0, -1) {};
\node(C) [vrtx= center/$\sqrt{2}$] at (1.5, 0) {};
\node(D) [vrtx= center/$15$] at (3, 0) {};

\path   (A) edge (B)
        (A) edge (C)
        (C) edge (B)
        (C) edge (D);
\end{tikzpicture}
    \caption*{Valuation $f$}  
    &   
\begin{tikzpicture}[
every edge/.style = {draw=black,very thick},
 vrtx/.style args = {#1/#2}{%
      circle, draw, thick, fill=white,
      minimum size=8mm, label=#1:#2}
                    ]
\node(A) [vrtx= center/$1$] at (0, 1) {};
\node(B) [vrtx= center/$2$] at (0, -1) {};
\node(C) [vrtx= center/$3$] at (1.5, 0) {};
\node(D) [vrtx= center/$4$] at (3, 0) {};

\path   (A) edge (B)
        (A) edge (C)
        (C) edge (B)
        (C) edge (D);
\end{tikzpicture}
\caption*{Valuation $g$}  
    \end{tabularx}
\caption{Two isomorphic valuations}
\label{graphs_ex_2}
\end{figure}

\begin{notation}
    The $\Box$ relation will be used to denote one of $<,=$ or $>$. This is to ease notation instead of writing (resp. $<, =, >$) in each proof. This is because condition 2 in Definition \ref{iso_of_val_def} requires us to check $<,=$ and $>$ but the argument is generally the same in each case.
\end{notation}
The definition below is more general than the one used in the paper.
\begin{definition}
    Let $G$ be a graph. The \textit{valuation family of $G$} is the set of all valuations of $G$ and is denoted by
    $$\mathcal{F}_G = \{f \, | f: V \to C \text{ is a valuation }\}$$
    When the graph is clear from context, $\mathcal{F}_G = \mathcal{F}$.
\end{definition}

 The isomorphism of valuations will allow us to define an equivalence relation on the valuation family $\mathcal{F}_G$. Provided the updates of the asynchronous maximum model are in some sense the same on these equivalence classes then the number of valuations can instead be bounded by the number of equivalence classes. This idea is made precise in the following three lemmas.

\begin{lemma}\label{iso_of_val_is_equiv_rel}
    Fix a graph $G$ and consider $\mathcal{F}$. Define the relation $f \equiv g$ if there exists an isomorphism of valuations from $f$ to $g$. Then $\equiv$ is an equivalence relation on $\mathcal{F}$.
\end{lemma}

\begin{proof}
    Let $C,C',C''$ be configurations. Let $f:V \to C,g:V \to C',h:V\to C''$ be valuations in $\mathcal{F}$.
    \par \textbf{Reflexive} The map $\alpha : V \to V$ given by $\alpha(v) = v$ is a graph isomorphism. Further, it is an isomorphism of valuations since for $u,v \in V$ we have $f(u) \Box f(v)$ if and only if $f(u) = f(\alpha(u)) \Box f(\alpha(v)) = f(v)$. Therefore $f \equiv f$.
    \par \textbf{Symmetric}
    Suppose $f \equiv g$. Then there exists an isomorphism of valuations $\alpha: V \to V$. We want to show $\alpha^{-1}$ is an isomorphism of valuations from $g$ to $f$. $\alpha^{-1}$ is a graph isomorphism since the inverse of a graph automorphism is itself a graph isomorphism.
    \par To show the second condition let $u',v' \in V$. Note these can be written as $u' = \alpha^{-1}(u)$ and $ v' = \alpha^{-1}(v))$ for some $u,v \in V$. Since $f\equiv g$, we have $f(u') \Box f(v')$ if and only if $g(\alpha(u')) \Box g(\alpha(v'))$. This implies $f(\alpha^{-1}(u)) \Box f(\alpha^{-1}(v))$ if and only if $g(u) = g(\alpha(\alpha^{-1}(u))) < \Box g(\alpha(\alpha^{-1}(u))) = g(v)$, which is exactly the second condition. Therefore $g \equiv f$.
    \par \textbf{Transitive}
    Suppose $f \equiv g$ under the isomorphism $\alpha$ and $g \equiv h$ under the isomorphism $\beta$. We show $f \equiv h$ under the isomorphism of valuations $\beta \circ \alpha$. A composition of automorphisms is an automorphism, so $\beta \circ \alpha$ indeed satisfies the first condition.
    \par Let $u,v \in V$. Since $f \equiv g$, we have $f(u) \Box f(v)$ if and only if $g(\alpha(u)) \Box g(\alpha(v))$. Further since $g \equiv h$, $g(\alpha(u)) \Box g(\alpha(v))$ if and only if $h(\beta(\alpha(u)) \Box h(\beta(\alpha(v)))$. Therefore
    $$f(u) \Box f(v) \text{ if and only if } h(\beta(\alpha(u)) \Box h(\beta(\alpha(v)))$$
    Therefore the relation is transitive.
\end{proof}

\begin{notation}
    Let $[\mathcal{F}]$ denote the equivalence classes of $\mathcal{F}$ under the equivalence relation $\equiv$ defined in $\ref{iso_of_val_is_equiv_rel}$. If $f \in \mathcal{F}$, let $[f]$ denote its equivalence class.
\end{notation}
An equivalence relation partitions a set into equivalence classes. We now show that valuations in the same equivalence class remain in the same equivalence class after being updated by the asynchronous maximum model. Informally the updates of the asynchronous maximum model do not care about the underlying configuration set $C$, but rather only the equivalence class of the valuation function $f: V \to C$.

\begin{lemma}\label{iso_con_period}
    Let $G$ be a graph and $C,C'$ be configurations. Let $f_t: V \to C$, $g_t: V \to C'$ be valuations such that $f_t \equiv g_t$ under the isomorphism $\alpha$. Then $f_{t+1} \equiv g_{t+1}$.
\end{lemma}
\begin{proof}
       We claim that $\alpha$ is also an isomorphism from $f_{t+1}$ to $g_{t+1}$. It is still a graph automorphism so it remains to show that the second condition is satisfied.
       Let $v' \in V$ be the vertex chosen for update under the asynchronous maximum model under valuation $f_t$. This corresponds to the choice of $\alpha(v')$ for the valuation $g_t$. By definition we have that
               \begin{equation*}
            f_{t+1}(v) = \begin{cases}
                \text{max}_{v \sim u}\{f_{t}(u)\} & \text{if } v = v' \\
                f_{t}(v) & \text{if } v \neq v' \\
            \end{cases}
        \end{equation*}

                \begin{equation*}
            g_{t+1}(\alpha(v)) = \begin{cases}
                \text{max}_{\alpha(v) \sim \alpha(u)}\{g_{t}(\alpha(u))\} & \text{if } \alpha(v) = \alpha(v') \\
                g_{t}(\alpha(v)) & \text{if } \alpha(v) \neq \alpha(v') \\
            \end{cases}
        \end{equation*}

    Let $u,v \in V$. We consider two cases.
    \par \textbf{Case 1: } Both $u,v \neq v'$. Since $\alpha$ is an isomorphism of valuations then for $u,v \in V$,
    
    \begin{equation*}
    f_t(u) \Box f_t(v) \text{ if and only if } g_t(\alpha(u)) \Box g_t(\alpha(v)) 
    \end{equation*}
    
    Since for $v\neq v'$, $f_{t+1}(v) = f_t(v)$ and $g_{t+1}(v) = g_t(v)$ then
    \begin{equation*}
     f_{t+1}(u) \Box f_{t+1}(v) \text{ if and only if } g_{t+1}(\alpha(u)) \Box g_{t+1}(\alpha(v)) 
    \end{equation*}
    Therefore the second condition of an isomorphism of valuations holds.
    
    \par \textbf{Case 2: } One of $u$ or $v$ is $v'$. Without loss of generality suppose $v = v'$. Suppose
    $$f_{t+1}(u) \Box f_{t+1}(v').$$
    We show this holds if and only if 
    $$g_{t+1}(\alpha(u)) \Box f_{t+1}(\alpha(v')).$$
    If $f_{t+1}(u) \Box f_{t+1}(v')$ then using the definition of $f_t$ gives
    \begin{equation}\label{assume}
        f_{t+1}(u) = f_{t}(u) \Box \max_{v'\sim w}\{f_t(w)\} = f_{t+1}(v')
    \end{equation}
    
    Let $w'$ be the vertex adjacent to $v'$ such that $\max_{v'\sim w}\{f_t(w)\} = f_t(w')$. Then Equation $\ref{assume}$ implies
    \begin{equation}\label{assume2}
        f_{t}(u) \Box f_t(w')
    \end{equation}
    
    Since $f_t \equiv g_t$, Equation \ref{assume2} holds \textit{if and only if} $g_{t}(\alpha(u))\Box g_t(\alpha(w')).$ But $g_{t}(\alpha(u)) = g_{t+1}(\alpha(u))$ since the vertex $\alpha(u)$ is not updated. Note $\alpha$ is a graph automorphism, so it preserves edge relations. Therefore $\alpha(v')\sim \alpha(w)$ if and only if $v'\sim w$. Further, $\alpha$ preserves in(equalities) between all vertices in the graph, so it preserves the maximum. Therefore
    $$g_t(\alpha(w')) = \max_{\alpha(v')\sim \alpha(w)}\{g_t(\alpha(w))\}$$
    But by definition we have
    $$g_t(\alpha(w')) = \max_{\alpha(v')\sim \alpha(w)}\{g_t(\alpha(w))\} = g_{t+1}(\alpha(v')).$$
    Combining this gives that
    $g_{t+1}(\alpha(u)) = g_{t}(\alpha(u))\Box g_t(\alpha(w')) = g_{t+1}(\alpha(v'))$. Therefore the second condition holds.
\end{proof}

Lemma \ref{iso_con_period} states we can consider updates on an element in $[\mathcal{F}]$ rather than on a single valuation in $\mathcal{F}$. When combined with the Lemma below it classifies the equivalence classes of possible valuations on an arbitrary graph $G$.

\begin{lemma}\label{iso_to_n}
    Let $G$ be a graph and $f: V \to C$ a valuation. There exists $g: V \to [n]$ and an isomorphism of valuations $\alpha$ from $f$ to $g$.
\end{lemma}
Therefore we can relabel the vertices with integers in $[n].$ $\alpha$ in Example $\ref{iso_of_val_ex}$ is one such isomorphism.

\begin{proof}
    Let $n = |V|$. Consider the valuation $f$. Let $v_1,v_2,\dots,v_n$ be the vertices of $V$. Since $C$ is a totally ordered set, there exists an ordering $i_1,i_2,\dots,i_n \in \N$ of the vertices such that the sequence
    $$f(v_{i_1}),f(v_{i_2}),f(v_{i_3}),\dots,f(v_{i_n})$$
    is non-decreasing. For $j \in \{1,2,\dots,n\}$, Define 
    \begin{equation*}
        g(v_{i_j}) = \begin{cases}
            1 & \text{ if $j = 1$} \\
            g(v_{i_{j-1}}) & \text{ if $j \geq 2$ and $f(v_{i_{j-1}}) = f(v_{i_{j-1}})$} \\
            g(v_{i_{j-1}}) + 1 & \text{ if $j \geq 2$ and $f(v_{i_{j-1}}) < f(v_{i_{j-1}})$} \\
        \end{cases}
    \end{equation*}
    Now consider $\alpha: V \to V$ which is the identity map, that is it maps every vertex to itself. This is an autormophism. Further, it preserves the ordering of the vertices since that is precisely how $g$ was defined, to preserve the order of the valuation $f$.
    
\end{proof}

\par 
Therefore Lemmas \ref{iso_con_period} and \ref{iso_to_n} state that every valuation on a graph $G$ can be though of as a function from $V \to [n]$. Then without loss of generality in the asynchronous maximum model we can assume $C = [n]$. Further, the number of equivalence classes in $[\mathcal{F}]$ is upper bounded by the number of functions from $V \to [n]$, which is $n^n$.

\section{Inequalities}
We will make use of some important inequalities.

\subsection{Harmonic Number}
\begin{lemma}\label{harmonic_num}
    Let $H_n = \sum_{i=1}^n \frac{1}{i}$ denote the $n$-th harmonic number. Then
    $$n H_{n-2} = \Theta(n \log n)$$
\end{lemma}
\begin{proof}
    It is well known (\cite{Boas}) that
    $$H_n = \log n + \gamma + \frac{1}{2n} - \varepsilon_n.$$
    Where $\gamma \approx 0.577 $ is a constant and $0\leq \varepsilon_n \leq \frac{1}{8n^2}$. First we show $nH_n = \Omega(n\log n)$. We have
    \begin{align*}
        & nH_n = n\log n + n\gamma + \frac{1}{2} - n\varepsilon_n && \text{} \\
        & \geq n\log n - \frac{n}{8n^2} && \text{} \\
        & = n\log n - \frac{1}{8n} && \text{} \\
        & \geq n\log n - \frac{1}{8} && \text{} \\
        & \geq n\log n(1-\frac{1}{8n\log n})&& \text{For $n \geq 2$} \\
        & \geq n\log n(1-\frac{1}{8 \cdot 2 \log 2})&& \text{For $n \geq 2$} \\
    \end{align*}
  Therefore $nH_{n} = \Omega(n\log n)$. Then we note that
  $$nH_n - \frac{n}{n-1} - 1 = nH_{n-2}$$
  For $n \geq 3$ we have
  $$nH_n - \frac{5}{2} \leq nH_{n-2}.$$
  Since $nH_{n} = \Omega(n\log n)$ then there exists a constant $C = (1-\frac{1}{8 \cdot 2 \log 2}) > 0$ such that
  $$C n\log n -\frac{5}{2} \leq nH_{n-2}.$$
  Therefore
  $$n\log n(C -\frac{5}{2 n\log n} )\leq nH_{n-2}.$$
  Let $N \geq 3$ be the smallest natural number such that $(C -\frac{5}{2 N\log N} ) > 0$. This occurs at $N = 3$. Then for $n \geq N = 3$ we have
   $$n\log n(C -\frac{5}{2 \cdot 3\log 3} )\leq nH_{n-2}.$$
   
   Therefore $n H_{n-2} = \Omega(n\log n)$. We note that for $n \geq 10$ we have that
   $$0.868 \cdot n\log n \approx  n\log n(C -\frac{5}{2 \cdot 10\log 10} ) \leq n\log n(C -\frac{5}{2 n\log n} )\leq nH_{n-2}$$
   Therefore even for small $n \geq 10$ we have a good lower bound on $nH_{n-2}$.

   \par Now we show $nH_{n-2} =  O(n \log n)$. We have that
   $$H_n = \log n + \gamma + \frac{1}{2n} - \varepsilon_n \leq \log n + 1 + \frac{1}{2} \leq \log n + \log 5.$$
   Therefore for $n \geq 5$, $H_n \leq 2 \log n$. So for $n \geq 5$ we have
   $$n H_{n-2} \leq n H_{n} \leq 2 n \log n = \bigO(n \log n) .$$
   Therefore $nH_{n-2} = \Theta(n\log n)$.
\end{proof}

\subsection{Vertex Boundaries Proof}
\begin{proof}\label{vertex_bound_proof}
Let $A \subset V$ be such that $0 < |A| \leq \frac{n}{2}$. Since $G$ is connected we have that $\frac{1}{\frac{n}{2}} \leq \frac{|\Gamma(A)|}{|A|}$. Therefore $\frac{2}{n} \leq \phi_{out}$.

    \par Now we show $\phi \leq 5$. For $n = 2$ vertices there is one connected graph on two vertices so we know $\phi = 1$.
    \par Assume $n\geq 3$. Let $A \subset V$ be any set of vertices such that $|A| = \lfloor \frac{n}{2} \rfloor$. Then
    \begin{align*}
        & \frac{|\Gamma(A)|}{|A|} = \frac{|\Gamma(A)|}{\lfloor \frac{n}{2} \rfloor} && \text{} \\
        & = \frac{|\Gamma(A)|}{\frac{n}{2}-1} && \text{} \\
        & = \frac{n - \lfloor \frac{n}{2} \rfloor}{\frac{n}{2}-1} && \text{$|A| = \lfloor \frac{n}{2} \rfloor$ so $|\Gamma(A)| \leq n - \lfloor \frac{n}{2} \rfloor$} \\
        & = \frac{n - (\frac{n}{2}-1)}{\frac{n}{2}-1} && \text{} \\
        & = \frac{\frac{n}{2}+1}{\frac{n}{2}-1} && \text{} \\
        & \leq 5 && \text{} \\
    \end{align*}
The last inequality is since the function $\frac{\frac{n}{2}+1}{\frac{n}{2}-1} $ is decreasing with respect to $n$ so is maximised at its endpoint $n = 3$.
Since $\phi_{out}$ is a minimum of $\frac{|\Gamma(A)|}{|A|}$ over all possible $A \subset V$ with $0 < |A| \leq \frac{n}{2}$ and $\frac{|\Gamma(A)|}{|A|} \leq 5$ for some choice of $A$ then we have $\phi_{out} \leq 5$.
\end{proof}

\section{Boundaries Proof}\label{n_boundaries_app}
\begin{lemma}
    Let $G$ be a strongly connected directed graph. Let $\emptyset \subsetneq S \subseteq  V$. There exists a $N \in \N \cup \{0\}$ such that $$V =  \bigsqcup_{k=0}^N \Gamma^k(S) $$
    where $\Gamma^k(S) \neq \emptyset$ for all $0 \leq k \leq N$ and $\Gamma^{k}(S) = \emptyset$ for all $k > N$. That is the $k$-boundaries of $S$ partition the set $V$ into exactly $N+1$ disjoint subsets.
\end{lemma}
\begin{proof}

We use induction on $|V|$.
    \par If $|V| = 1$ then $S = V$ by assumption, since $S\subseteq V$ is non-empty. Then we have
    $$V = S.$$
    is a partition into $k$-boundaries with $N = 0$.
    \par Now suppose $|V| > 1$. The idea of the proof is to choose a suitable vertex $v \in V$ to delete from $G$. Denote $V' = V \setminus \{v\}$. This results in $G' = G[V']$. The induced subgraph will have $|V| - 1$ vertices and by the induction hypothesis $V'$ is partitioned into $k$-boundaries by any $S' \subseteq V'$. Then we show adding back in this vertex $v$ does not violate any of the desired properties.

    \par 

\begin{figure}[!h]
    \centering
    \includegraphics[scale = 0.4]{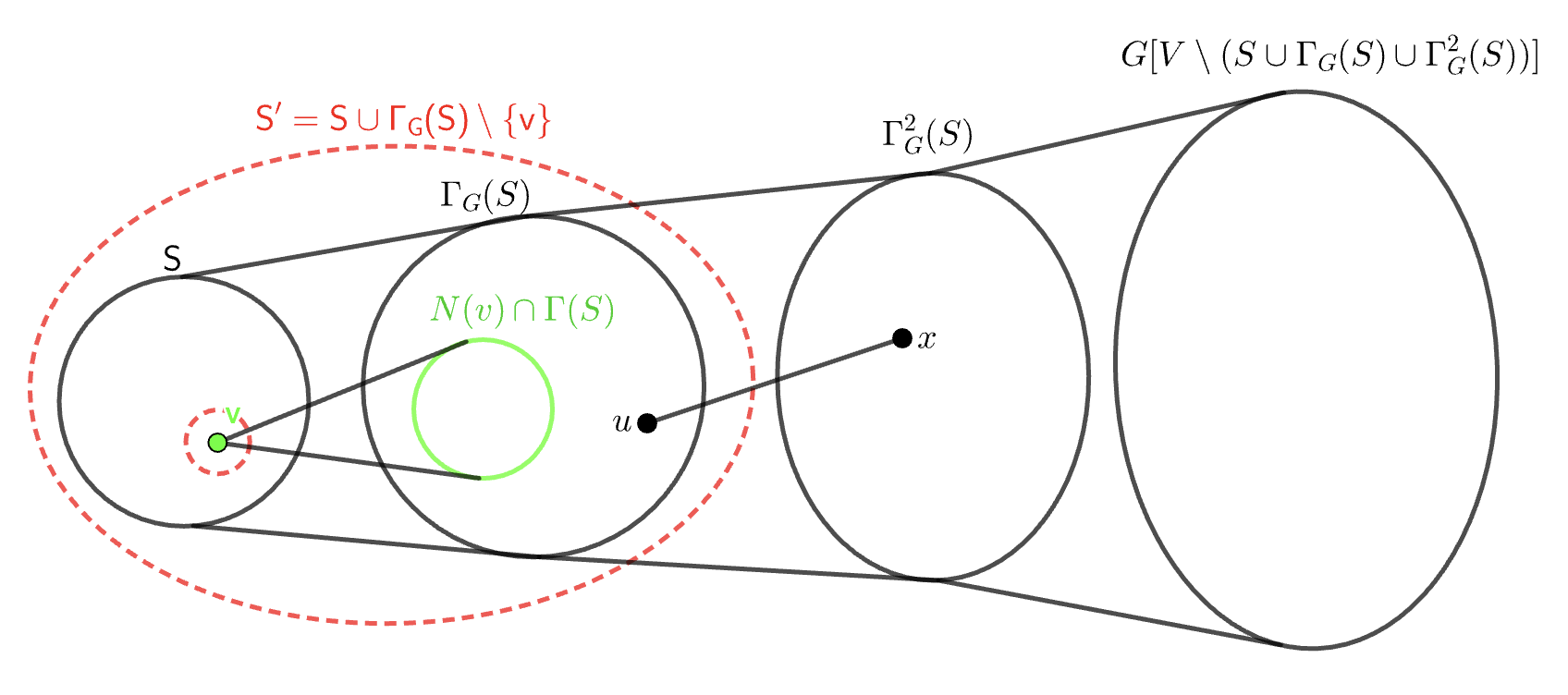}
    \caption{Graph $G$ with components $S, \Gamma(S)$, $S'$,}
    \label{fig:nboundproof}
\end{figure}

Please refer to Figure \ref{fig:nboundproof}. Let $S \subseteq V$. Since $S$ is non-empty then we can fix any $v \in S$. Consider $V' = V \setminus \{v\}$ and let $G' = G[V'] = (V',E')$ be the subgraph of $G$ induced by $V'$. Since $|V'| = |V| - 1$ the induction hypothesis implies for any non empty set $S'$, there exists an $N \in \N$ such that
\begin{equation}\label{partitionV'}
    V' = \bigsqcup_{k=0}^N \Gamma_{G'}^k(S').
\end{equation}

where for all $0\leq k \leq N$, $\Gamma_{G'}^k(S') \neq \emptyset$ and $\Gamma_{G'}^k(S') = \emptyset$ for all $k > N$. Take $S' = (S \cup \Gamma_G(S)) \setminus \{v\}$ in the induction hypothesis.
\par \textbf{Claim:} We show that for $k = 1,2,\dots, N$,
$$\Gamma_{G'}^k(S') = \Gamma_{G}^{k+1}(S).$$

The proof of the claim is by induction on $k$.
\par \textbf{Base Case:}
For $k = 1$ this amounts to using the definition of the boundary. We show $\Gamma_{G'}(S') \subseteq \Gamma_{G}^{2}(S)$ and $\Gamma_{G'}(S') \supseteq \Gamma_{G}^{2}(S)$.
\par Suppose $x \in \Gamma_{G'}(S') $. Then $x \in V'$ and there exists an edge $(u,x) \in E'$ such that
\begin{equation}\label{uinS'}
    u \in S' = S \cup \Gamma_{G}(S) \setminus \{v\} 
\end{equation}
\begin{equation} \label{xnotinS'}
    x \notin S' = S \cup \Gamma_{G}(S) \setminus \{v\}
\end{equation}

To show that $x \in \Gamma_G^2(S)$ we use the definition,
\begin{align*}
    & \Gamma_G^2(S) = \Gamma_G(\Gamma_G(S)) \setminus S  && \text{}\\
    & = \{v' \in V | u \in \Gamma_G(S), v' \notin \Gamma_G(S) \text{ and } u \sim v'\} \setminus S  && \text{}\\
\end{align*}
Which implies
\begin{equation}\label{gamma^2}
   \Gamma_G^2(S)  = \{v' \in V | u \in \Gamma_G(S), v' \notin \Gamma_G(S) \cup S \text{ and } u \sim v'\} 
\end{equation}

Observe (\ref{uinS'}), (\ref{xnotinS'}) and that $(u,x)$ is an edge in $G'$ implies $u \in \Gamma_{G}(S)$. This is because the only edges between the sets $S$ and $(S')^c$ contain the vertex $v$. However $N(v) \subseteq S'$ which would imply $x \in S'$ also, a contradiction. So the edge from $u$ to $x$ must be from $ u \in \Gamma_{G}(S)$ to $x \notin S'$.  Therefore $x \in \Gamma_G^2(S)$ by (\ref{gamma^2}). This shows $\Gamma_{G'}(S') \subseteq \Gamma_{G}^{2}(S)$.
\par Now suppose $x \in \Gamma_G^2(S)$. Then there exists an edge $(u,x) \in E$ such that
\begin{equation*}
    u \in \Gamma_G(S) = \{v' \in V | w \in S, v' \notin S \text{ and } w \sim v'\}
\end{equation*}
\begin{equation} \label{xnotingammaScupS}
    x \notin \Gamma_G(S) \cup S = \{v' \in V | w \in S, v' \notin S \text{ and } w \sim v'\} \cup S
\end{equation}
Since $u \in \Gamma_G(S)$ then $u \in S'$ (satisfying (\ref{uinS'})). Further, $x \notin S'$ (the only check we need is that $x \neq v$, which holds by (\ref{xnotingammaScupS}) since $x \notin S$.)
Then we have shown the edge $(u,x) \in E' \subseteq E$ and that $x \in V'$. Therefore $x \in \Gamma_{G'}(S')$. This shows $\Gamma_{G'}(S') \supseteq \Gamma_{G}^{2}(S)$.
Therefore we have that $\Gamma_{G'}(S') = \Gamma_{G}^{2}(S)$

\par \textbf{Induction:} Let $k \geq 2 $ in the claim. Then
\begin{align*}
    & \Gamma_{G'}^k(S') = \Gamma_{G'}(\Gamma_{G'}^{k-1}(S')) \setminus \Gamma_{G'}^{k-2}(S')&& \text{Definition of $\Gamma_{G'}(S')$} \\
    & = \Gamma_{G'}(\Gamma_{G}^{k}(S)) \setminus \Gamma_{G}^{k-1}(S)&& \text{Inductive hypothesis} \\
    & = \Gamma_{G}(\Gamma_{G}^{k}(S)) \setminus \Gamma_{G}^{k-1}(S)&& \text{See *} \\
    & = \Gamma_{G}^{k+1}(S)&& \text{Definition} \\
\end{align*}
(*) Let $A\subset V$ be any set such that $A \cap (N(v) \cup \{v\}) = \emptyset$. That is $A$ does not contain $v$ or any vertices in the neighbourhood. $\Gamma_G(A) = \Gamma_{G'}(A)$ since any vertex in one set satisfies the conditions of being in the other.
$$\Gamma_G(A) = \{v \in V | u \in A, v \notin A \text{ and } u \sim v\}  = \{v' \in V' | u \in A, v' \notin A \text{ and } u \sim v'\} = \Gamma_{G'}(A).$$
Letting $A = \Gamma_G^k(S)$ and noting that since $k\geq 2$, $A \cap (N(v) \cup \{v\}) = \emptyset$.
\newline
\par 
Therefore the claim holds by finite induction.
\par Now consider Equation $\ref{partitionV'}$. We can write this as
\begin{align*}
    & V' = \bigsqcup_{k=0}^N \Gamma_{G'}^k(S') && \text{} \\
    & = S' \sqcup \Gamma_{G'}(S') \sqcup \Gamma_{G'}^2(S') \sqcup \dots \sqcup \Gamma_{G'}^N(S') && \text{} \\
    & = S' \sqcup \Gamma_{G}(S) \sqcup \Gamma_{G}^{2}(S) \sqcup \dots \sqcup \Gamma_{G}^{N+1}(S) && \text{By Claim} \\
    & = (S \cup \Gamma(S) \setminus\{v\})  \sqcup \Gamma_{G}(S) \sqcup \Gamma_{G}^{2}(S) \sqcup \dots \sqcup \Gamma_{G}^{N+1}(S) && \text{By definition of $S'$} \\
\end{align*}
Therefore, since $S$ and $\Gamma(S)$ are disjoint by definition and $v$ is only a member of $S$, then
$$V = V' \sqcup \{v\} = S \sqcup \Gamma(S)  \sqcup \Gamma_{G}(S) \sqcup \Gamma_{G}^{2}(S) \sqcup \dots \sqcup \Gamma_{G}^{N+1}(S). $$
Which is the desired partition of $V$.

\end{proof}

\section{Coupling Procedure} \label{coupling_proc_app}
In this section we provide a `potential proof' of Theorem \ref{coupling_proof}. Please note this is not a complete proof and may require more attention. For clarity the theorem is stated below. The general idea can be found in \cite{couple_arg}.

\begin{theorem*}
Let $Q_j$ for $j\geq 1$ be a Bernoulli random variable with success probability $q \neq 1$. That is
    $$Q_j = \begin{cases}
        1 & \text{with probability } q \\
        0 & \text{with probability } 1-q 
    \end{cases}$$
    Suppose that for all $j \geq 1$, $q\leq p_j \leq 1$. Let $P_j$ for $j\geq 1$ be a random variable defined by
    $$P_j = \begin{cases}
        1 & \text{with probability } p_j \\
        0 & \text{with probability } 1-p_j \\
    \end{cases}$$
    Define the following random variables.
    $$Q = \min\{j | Q_j = 1\}$$
    $$P = \min\{j | P_j = 1\}$$
    Then for all $k \in \N$, $\mathbb{P}(Q \leq k) \leq \mathbb{P}(P \leq k)$. In other words, $Q$ is stochastically dominated by $P$.
\end{theorem*}
\begin{proof}
    We define the following procedure called \textit{Coupling}$(P,Q)$. This procedure takes two random variables $P$ and $Q$ with the assumptions given in the theorem as input. The output is a new random variable $P'$ which has the same distribution as $P$ but is `coupled' with $Q$. Given $P,Q$ (and therefore $Q_j,P_j$) define
    $$P_j' = \begin{cases}
        1 & \text{if $Q_j = 1$} \\
        1 & \text{with probability $\frac{p_j-q}{1-q}$} \\
        0 & \text{otherwise} \\
    \end{cases}$$
That is we imitate $Q_j$ whenever it is a success, and even if $Q_j$ is not a success, $P_j'$ may still be a success with some well defined probability. It is well defined since $q\neq 1$ and $p_j \geq q$. Further define the output of \textit{Coupling}$(P,Q)$ as
$$P' = \min\{j | P_j' = 1\}.$$
We have that
$$\mathbb{P}(P_j' = 1) = q + (1-q)(\frac{p_j-q}{1-q}) = p_j.$$
Therefore the distribution of $P_j'$ is the same as the distribution of $P_j$. This also implies $P$ and $P'$ have the same distribution. If we treat $1$ as a success and $0$ as a failure, the process defined by $P_j'$ always reaches a success before or at the same round as the process $Q_j$.

Therefore 
$$\mathbb{P}(Q \leq k) \leq \mathbb{P}(P' \leq k) = \mathbb{P}(P \leq k).$$
where the last equality is since $P$ and $P'$ have the same distribution. Therefore $P$ stochastically dominates $Q$ as desired.
    
\end{proof}

\section{Large Figures}
Below we introduce a large figure. It displays 7 rounds of the asynchronous maximum model where the highlighted vertices are chosen for the update.

\begin{figure}[h!]
\centering
\begin{tikzpicture}[
every edge/.style = {->, draw=black,very thick},
 vrtx/.style args = {#1/#2}{%
      circle, draw, thick, fill=white,
      minimum size=8mm, label=#1:#2}
                    ]
\node(A) [vrtx= center/6] at (0, 0) {};
\node(B) [vrtx= center/5] at (2, 0) {};
\node(C) [vrtx= center/4] at (3.2, -2) {};
\node(D) [vrtx= center/3] at (2, -4) {};
\node(E) [vrtx= center/2] at (0, -4) {};
\node(F) [vrtx= center/1] at (-1.2, -2) {};
\node(G) [vrtx= center/5] at (5.2, -2) {};
\node(H) [vrtx= center/2] at (7.2, -0.7) {};
\node(I) [vrtx= center/3] at (7.2, -3.3) {};
\path   (A) edge (B)
        (B) edge (C)
        (C) edge (D)
        (D) edge (E)
        (E) edge (F)
        (F) edge (A)
        (C) edge (G)
        (H) edge (G)
        (I) edge (H)
        (G) edge (I)
        (G) edge (D);
\end{tikzpicture}\clearpage
    \caption*{Round $0$: $h(f_0) = 0$}  
\begin{tikzpicture}[
every edge/.style = {->, draw=black,very thick},
 vrtx/.style args = {#1/#2}{%
      circle, draw, thick, fill=white,
      minimum size=8mm, label=#1:#2}
                    ]
\node(A) [red][vrtx= center/5] at (0, 0) {};
\node(B) [vrtx= center/5] at (2, 0) {};
\node(C) [vrtx= center/4] at (3.2, -2) {};
\node(D) [vrtx= center/3] at (2, -4) {};
\node(E) [vrtx= center/2] at (0, -4) {};
\node(F) [vrtx= center/1] at (-1.2, -2) {};
\node(G) [vrtx= center/5] at (5.2, -2) {};
\node(H) [vrtx= center/2] at (7.2, -0.7) {};
\node(I) [vrtx= center/3] at (7.2, -3.3) {};
\path   (A) edge (B)
        (B) edge (C)
        (C) edge (D)
        (D) edge (E)
        (E) edge (F)
        (F) edge (A)
        (C) edge (G)
        (H) edge (G)
        (I) edge (H)
        (G) edge (I)
        (G) edge (D);
\end{tikzpicture}
    \caption*{Round $1$: $h(f_1) = 0$}  

\begin{tikzpicture}[
every edge/.style = {->, draw=black,very thick},
 vrtx/.style args = {#1/#2}{%
      circle, draw, thick, fill=white,
      minimum size=8mm, label=#1:#2}
                    ]
\node(A) [vrtx= center/5] at (0, 0) {};
\node(B) [vrtx= center/5] at (2, 0) {};
\node(C) [red][vrtx= center/5] at (3.2, -2) {};
\node(D) [vrtx= center/3] at (2, -4) {};
\node(E) [vrtx= center/2] at (0, -4) {};
\node(F) [vrtx= center/1] at (-1.2, -2) {};
\node(G) [vrtx= center/5] at (5.2, -2) {};
\node(H) [vrtx= center/2] at (7.2, -0.7) {};
\node(I) [vrtx= center/3] at (7.2, -3.3) {};
\path   (A) edge (B)
        (B) edge (C)
        (C) edge (D)
        (D) edge (E)
        (E) edge (F)
        (F) edge (A)
        (C) edge (G)
        (H) edge (G)
        (I) edge (H)
        (G) edge (I)
        (G) edge (D);
\end{tikzpicture}
    \caption*{Round $2$: $h(f_2) = 0$}  
\begin{tikzpicture}[
every edge/.style = {->, draw=black,very thick},
 vrtx/.style args = {#1/#2}{%
      circle, draw, thick, fill=white,
      minimum size=8mm, label=#1:#2}
                    ]
\node(A) [vrtx= center/5] at (0, 0) {};
\node(B) [vrtx= center/5] at (2, 0) {};
\node(C) [vrtx= center/5] at (3.2, -2) {};
\node(D) [vrtx= center/3] at (2, -4) {};
\node(E) [vrtx= center/2] at (0, -4) {};
\node(F) [vrtx= center/1] at (-1.2, -2) {};
\node(G) [vrtx= center/5] at (5.2, -2) {};
\node(H) [red][vrtx= center/5] at (7.2, -0.7) {};
\node(I) [vrtx= center/3] at (7.2, -3.3) {};
\path   (A) edge (B)
        (B) edge (C)
        (C) edge (D)
        (D) edge (E)
        (E) edge (F)
        (F) edge (A)
        (C) edge (G)
        (H) edge (G)
        (I) edge (H)
        (G) edge (I)
        (G) edge (D);
\end{tikzpicture}
    \caption*{Round $3$: $h(f_3) = 0$}  
\caption{Round 0 to Round 3}
\label{sc_example}
\end{figure} 

\begin{figure}[h!]
\centering
\begin{tikzpicture}[
every edge/.style = {->, draw=black,very thick},
 vrtx/.style args = {#1/#2}{%
      circle, draw, thick, fill=white,
      minimum size=8mm, label=#1:#2}
                    ]
\node(A) [vrtx= center/5] at (0, 0) {};
\node(B) [vrtx= center/5] at (2, 0) {};
\node(C) [vrtx= center/5] at (3.2, -2) {};
\node(D) [vrtx= center/3] at (2, -4) {};
\node(E) [vrtx= center/2] at (0, -4) {};
\node(F) [vrtx= center/1] at (-1.2, -2) {};
\node(G) [red][vrtx= center/3] at (5.2, -2) {};
\node(H) [vrtx= center/5] at (7.2, -0.7) {};
\node(I) [vrtx= center/3] at (7.2, -3.3) {};
\path   (A) edge (B)
        (B) edge (C)
        (C) edge (D)
        (D) edge (E)
        (E) edge (F)
        (F) edge (A)
        (C) edge (G)
        (H) edge (G)
        (I) edge (H)
        (G) edge (I)
        (G) edge (D);
\end{tikzpicture}\clearpage
    \caption*{Round $4$: $h(f_4) = 0$}  

\begin{tikzpicture}[
every edge/.style = {->, draw=black,very thick},
 vrtx/.style args = {#1/#2}{%
      circle, draw, thick, fill=white,
      minimum size=8mm, label=#1:#2}
                    ]
\node(A) [vrtx= center/5] at (0, 0) {};
\node(B) [vrtx= center/5] at (2, 0) {};
\node(C) [vrtx= center/5] at (3.2, -2) {};
\node(D) [vrtx= center/3] at (2, -4) {};
\node(E) [vrtx= center/2] at (0, -4) {};
\node(F) [vrtx= center/1] at (-1.2, -2) {};
\node(G) [vrtx= center/3] at (5.2, -2) {};
\node(H) [vrtx= center/5] at (7.2, -0.7) {};
\node(I) [red][vrtx= center/5] at (7.2, -3.3) {};
\path   (A) edge (B)
        (B) edge (C)
        (C) edge (D)
        (D) edge (E)
        (E) edge (F)
        (F) edge (A)
        (C) edge (G)
        (H) edge (G)
        (I) edge (H)
        (G) edge (I)
        (G) edge (D);
\end{tikzpicture}
    \caption*{Round $5$: $h(f_5) = 0$}  

\begin{tikzpicture}[
every edge/.style = {->, draw=black,very thick},
 vrtx/.style args = {#1/#2}{%
      circle, draw, thick, fill=white,
      minimum size=8mm, label=#1:#2}
                    ]
\node(A) [vrtx= center/5] at (0, 0) {};
\node(B) [vrtx= center/5] at (2, 0) {};
\node(C) [vrtx= center/5] at (3.2, -2) {};
\node(D) [vrtx= center/3] at (2, -4) {};
\node(E) [vrtx= center/2] at (0, -4) {};
\node(F) [vrtx= center/1] at (-1.2, -2) {};
\node(G) [red][vrtx= center/5] at (5.2, -2) {};
\node(H) [vrtx= center/5] at (7.2, -0.7) {};
\node(I) [vrtx= center/5] at (7.2, -3.3) {};
\path   (A) edge (B)
        (B) edge (C)
        (C) edge (D)
        (D) edge (E)
        (E) edge (F)
        (F) edge (A)
        (C) edge (G)
        (H) edge (G)
        (I) edge (H)
        (G) edge (I)
        (G) edge (D);
\end{tikzpicture}
    \caption*{Round $6$: $h(f_6) = 6$ and a strong cycle is formed}  
\begin{tikzpicture}[
every edge/.style = {->, draw=black,very thick},
 vrtx/.style args = {#1/#2}{%
      circle, draw, thick, fill=white,
      minimum size=8mm, label=#1:#2}
                    ]
\node(A) [vrtx= center/5] at (0, 0) {};
\node(B) [vrtx= center/5] at (2, 0) {};
\node(C) [vrtx= center/5] at (3.2, -2) {};
\node(D) [red][vrtx= center/2] at (2, -4) {};
\node(E) [vrtx= center/2] at (0, -4) {};
\node(F) [vrtx= center/1] at (-1.2, -2) {};
\node(G) [vrtx= center/5] at (5.2, -2) {};
\node(H) [vrtx= center/5] at (7.2, -0.7) {};
\node(I) [vrtx= center/5] at (7.2, -3.3) {};
\path   (A) edge (B)
        (B) edge (C)
        (C) edge (D)
        (D) edge (E)
        (E) edge (F)
        (F) edge (A)
        (C) edge (G)
        (H) edge (G)
        (I) edge (H)
        (G) edge (I)
        (G) edge (D);
\end{tikzpicture}
    \caption*{Round $7$: $h(f_7) = 6$}  
\caption{Round 4 to Round 7}
\end{figure}